\documentclass[journal, twoside]{IEEEtran}
\IEEEoverridecommandlockouts

\usepackage{textcomp} 
\usepackage{xspace}
\usepackage[normalem]{ulem} 
\usepackage{cite}
\usepackage{enumitem} 
\usepackage{fancyhdr}

\usepackage{pgfplots}
\usepackage[footnotesize]{caption}
\usepackage{subfig}
\usepackage{graphicx}
\usepackage{tabularx}
\usepackage{booktabs} 
\usepackage{circuitikz}
\usepackage{amsmath, amsbsy, amssymb, amsthm, amsfonts,bm}

\usepackage{mathtools} 
\usepackage{mathrsfs}
\usepackage{multirow}
\usepackage{gensymb} 
\usepackage{cancel} 
\usepackage{ifthen}
\usepackage{stmaryrd} 
\usepackage{dsfont} 

\newtheorem{theorem}{Theorem}
\newtheorem{lemma}[theorem]{Lemma}
\newtheorem{proposition}[theorem]{Proposition}
\newtheorem{corollary}{Corollary}

\newtheorem{assumption}{Assumption}

\theoremstyle{definition}

\newtheorem*{remark*}{Remark}

\usepackage{setspace}
\usepackage{lscape}
\usepackage{pdfpages} 
\usepackage{xcolor}
\pgfplotsset{compat=newest}
\pgfplotsset{plot coordinates/math parser=false}
\newlength\figureheight
\newlength\figurewidth
\pgfplotsset{compat=newest}
\pgfplotsset{plot coordinates/math parser=false}

\usepackage[hyphens]{url}
\usepackage[hidelinks]{hyperref}
	
\newcommand{\splev}{{K_0}}

\newcommand{\bb}{\mathbb}
\newcommand{\bs}{\boldsymbol}
\newcommand{\bc}[1]{\boldsymbol{\mathcal #1}}
\newcommand{\cl}{\mathcal}
\newcommand{\ol}{\overline}
\newcommand{\wt}{\widetilde}

\newcommand{\st}{%
  \ifmmode
  \ {\rm s.t.}\ %
  \else%
  \emph{s.t.}\@\xspace%
  \fi%
}
\newcommand{\whp}{\emph{w.h.p.}\xspace}
\newcommand{\ie}{\emph{i.e.}, }
\newcommand{\eg}{\emph{e.g.}, }
\newcommand{\etal}{\emph{et al.}}

\newcommand{\Rbb}{\bb{R}} 
\newcommand{\Cbb}{\bb{C}} 

\DeclareMathOperator{\Id}{{\rm \bf Id}} 
\newcommand{\tinv}[1]{{\textstyle\frac{1}{#1}}}

\newcommand{\im}{\mathrm{i}\mkern1mu} 
\newcommand{\real}[1]{\Re{\left\{#1\right\}}} 
\newcommand{\conj}[1]{#1^{*}} 

\newcommand{\proba}[2][]{
  {
    \ifx&#1& \bb{P}\left[#2\right] 
    \else\bb{P} #1[#2 #1]\fi
  }
}

\newcommand{\iid}{%
  \ifmmode
  \mathrm{i.i.d.}%
  \else%
  i.i.d.\@\xspace%
  \fi%
}

\newcommand{\ud}{\mathrm{d}} 
\newcommand{\supp}{{\rm supp}\,} 
\DeclareMathOperator{\tr}{tr} 
\DeclareMathOperator{\diag}{diag}
\newcommand{\dmat}[1]{#1_\text{d}} 
\newcommand{\transpose}[1]{#1^\top} 
\newcommand{\norm}[3][]{#1\lVert#2#1\rVert_{#3}} 
\newcommand{\scp}[3][]{#1\langle #2, #3 #1\rangle} 
\newcommand{\fro}[3][]{#1\langle #2, #3 #1\rangle_{\text{F}}} 

\DeclareMathOperator*{\argmin}{arg\,min}

\newcommand{\intMa}[1][]{
   \ifthenelse{ \equal{#1}{} }{\bs{\cl I}}{\cl I_{#1}}
}
\newcommand{\intM}[3][]{ \intMa[#2] #1[#3 #1] } 
\newcommand{\intMOm}[2][]{
   \ifthenelse{ \equal{#2}{} }{\bs{\cl I}_\Omega}{\bs{\cl I}_\Omega #1[ #2 #1]} 
}
\newcommand{\intMOmT}[2][]{
   \ifthenelse{ \equal{#2}{} }{\bs{\cl I}_\Omega^*}{\bs{\cl I}_\Omega^* #1[ #2 #1]} 
}
\newcommand{\dintMOm}[2][]{
   \ifthenelse{ \equal{#2}{} }{\tilde{\bs{\cl I}}_\Omega}{\tilde{\bs{\cl I}}_\Omega #1[ #2 #1]} 
}
\newcommand{\intCirc}[1][]{
   \ifthenelse{ \equal{#1}{} }{\bs{\cl J}}{\cl J_{#1}} 
}
\newcommand{\dintCirc}[1][]{
   \ifthenelse{ \equal{#1}{} }{\tilde{\bs{\cl J}}}{\tilde{\cl J}_{#1}} 
}
\newcommand{\ropA}{\bs{\cl A}} 
\newcommand{\ILEop}{\bs{\cl B}} 

\newcommand{\fvign}{f^{\circ}}
\newcommand{\ts}{\textstyle}

\renewcommand{\leq}{\leqslant}
\renewcommand{\geq}{\geqslant}

\makeatletter
\newcommand{\mathleft}{\@fleqntrue\@mathmargin0pt}
\newcommand{\mathcenter}{\@fleqnfalse}

\usepackage{hyperref}
\hypersetup{
    colorlinks=true,
    linkcolor=black,
    filecolor=magenta,      
    urlcolor=cyan,
    citecolor = magenta
}
\usepackage{cleveref}
\crefformat{equation}{(#2\color{blue}#1#3)} 

\definecolor{NavyBlue}{RGB}{47,113,194}
\definecolor{greenmat}{RGB}{0,136,43}
\definecolor{org}{RGB}{222,106,16}
\definecolor{greenmat2}{RGB}{0,166,43}
\definecolor{org2}{RGB}{240,100,16}
\definecolor{MH}{RGB}{160,160,0}
\definecolor{SS}{RGB}{200,60,120}

\newcommand{\replace}[2][\@empty]
{ \ifbool{replace}
  {
  #2
  }
  {
  \ifx\@empty#1\relax \textcolor{greenmat2}{[BY]: #2} 
  \else\textcolor{org2}{[REPLACE]: #1} \\ \textcolor{greenmat2}{[BY]: #2}\fi
  }
}
\newcommand{\br}[1]{\textcolor{red}{[\textbf{BR:#1}]}}
\renewcommand{\br}[1]{}

\title{Interferometric Lensless Imaging:\\[1mm]
  \LARGE Rank-one Projections of Image Frequencies with Speckle Illuminations
}

\author{Olivier Leblanc\IEEEauthorrefmark{1}, Matthias Hofer\IEEEauthorrefmark{2}, 
Siddharth Sivankutty\IEEEauthorrefmark{3}, Hervé 
Rigneault\IEEEauthorrefmark{2}, and Laurent Jacques\IEEEauthorrefmark{1}
\thanks{\IEEEauthorrefmark{1} E-mail: {\em \{o.leblanc, 
laurent.jacques\}@uclouvain.be}. ISPGroup, INMA/ICTEAM, UCLouvain, 
Louvain-la-Neuve, Belgium. OL is funded by Belgian National 
Science Foundation (F.R.S.-FNRS).}
\thanks{\IEEEauthorrefmark{2}Aix Marseille Univ, CNRS, Centrale Méditerranée, Institut Fresnel, Marseille, France. \IEEEauthorrefmark{3}Univ. Lille, CNRS, UMR8523-PhLAM-Physique des Lasers, Atomes et Molécules, F-59000 Lille, France.}}

\begin{document}
\maketitle
\begin{abstract}
Lensless illumination single-pixel imaging with a multicore fiber (MCF) is a computational imaging technique that enables potential endoscopic observations of biological samples at cellular scale. In this work, we show that this technique is tantamount to collecting multiple symmetric rank-one projections (SROP) of an \emph{interferometric} matrix---a matrix encoding the spectral content of the sample image. In this model, each SROP is induced by the complex \emph{sketching} vector shaping the incident light wavefront with a spatial light modulator (SLM), while the projected interferometric matrix collects up to $\cl O(Q^2)$ image frequencies for a $Q$-core MCF. While this scheme subsumes previous sensing modalities, such as raster scanning (RS) imaging with beamformed illumination, we demonstrate that collecting the measurements of $M$ random SLM configurations---and thus acquiring $M$ SROPs---allows us to estimate an image of interest if $M$ and $Q$ scale log-linearly with the image sparsity level.
This demonstration is achieved both theoretically, with a specific restricted isometry analysis of the sensing scheme, and with extensive Monte Carlo experiments. 
On a practical side, we perform a single calibration of the sensing system robust to certain deviations to the theoretical model and independent of the $\emph{sketching}$ vectors used during the imaging phase. Experimental results made on an actual MCF system demonstrate the effectiveness of this imaging procedure on a benchmark image.
\end{abstract}


\section{Introduction}
\label{sec:intro}
The advent of Computational Imaging (CI) can be traced back to the work of Ables~\cite{ables1968fourier} and Dicke~\cite{dicke1968scatter} on coded aperture for x-ray and gamma ray imagers. Since then, an ever-growing number of solutions have been devised to relax the constraints imposed by more traditional optical architectures (when these exist). Cheaper, lighter, and enabling larger imaging field-of-view (FOV), Lensless Imaging (LI), a subfield of CI, is convenient for medical applications such as microscopy~\cite{microscopy2} and \emph{in vivo} imaging~\cite{invivo1} where the extreme miniaturization of the imaging probe (with a diameter $\le$ 200 $\micro$m) offers a minimally invasive route to image at depths unreachable in microscopy~\cite{Boominathan2016}.
More recently, intensive research effort emerged for Lensless Endoscopy (LE) using multimode~\cite{Lochocki22, Psaltis2016} or MultiCore Fibers (MCF)~\cite{Sivankutty2016, Andresen2016, Sun22}, paving the way for deep biological tissues~\cite{Choi22} and brain imaging. 

In CI applications, a mathematical model describes the observations as a function of the object to be imaged. Two efficiency requirements are considered; \emph{(i)} the model, while physically reliable, must be computationally efficient to speed up the reconstruction algorithms; \emph{(ii)} the acquisition method must minimize the number of observations (also called \emph{sample complexity}) needed to accurately estimate the object.  In single-pixel MCF-LI (see Fig.~\ref{fig:LIMCF}), \emph{Speckle Imaging} (SI) consists in randomly shaping the wavefront of the light input to the cores entering the MCF to illuminate the entire object with a randomly distributed intensity. As the MCF features a double cladding that collects and brings back a fraction of the light re-emitted (either at other wavelengths by fluorescence or by simple reflection) towards a high sensitive (single-pixel) detector, the light integrated by this sensor represents a complete projection of the speckle on the object.  Compared to Raster Scanning (RS) the object with a translating focused (beamformed) spot~\cite{Sivankutty2016}, SI reduces the overall sample complexity needed to estimate a reliable image~\cite{Guerit2021}.

In this work, we improve the MCF-LI sensing model jointly on its reliability, computation and calibration. We achieve this by introducing light propagation physics in the forward model of MCF imaging, while keeping the low sample complexity enabled by SI. Inserting the physics yields a sensing model similar to radio-inteferometry applications~\cite{wiaux2009compressed}, where the interferences of the light emitted by the cores composing the MCF give specific access to the Fourier content of the object to be imaged. 
The sample complexity of the underlying model is analyzed both theoretically and experimentally. 

\subsection{Related works} \label{sec:related}

In 2008, Duarte \etal~introduced single-pixel imaging~\cite{Duarte2008, Taylor22}, a subfield of lensless imaging (LI) where each collected observation is equivalent to randomly modulating an image before integrating its intensity. 
They demonstrated that reliable image estimation is possible at low sampling rates compared to image resolution by using compressive sensing. 
More recently, this principle has been integrated into the use of an MCF for both remote illumination and image collection. This technique allows for both deep and large FOV imaging~\cite{Sivankutty2016, Sivankutty2018,Guerit2021}. 
Subsequent works have shown that de-structured speckle-based illuminations can replace structured or beamformed illuminations effectively~\cite{Caravaca-Aguirre2018a, Guerit2021}.

MCF-LI bears similarities with quadratic measurement models such as \emph{phase retrieval} (PR)~\cite{Fienup82, Bauschke02} whose sensing is often recast as SROPs of the \emph{lifted} matrix $\bs x \bs x^*$ of the (vectorized) image $\bs x$. Theoretical guarantees on the recovery of low-complexity matrices (\eg sparse, circulant, low-rank) from random ROPs have been extensively studied in the last decade~\cite{chen2015exact, cai2015rop,Soltani2017}. 
Our sensing model computes SROPs of an \emph{interferometric} matrix built from spatial frequencies of the image. This shares similarities with \emph{random partial Fourier sensing} in compressive sensing (CS) theory~\cite{Candes2006a, foucart2017mathematical}. 
Speciﬁcally, the spatial frequencies in this matrix correspond to the difference of the MCF cores locations. This arises in radio-interferometric astronomy applications where, as induced by the van Cittert-Zernike theorem, the signal correlation of two antennas gives the Fourier content on a frequency vector (or \emph{visibility}) related to the baseline vector of the antenna pair~\cite{wiaux2009compressed}. One may recognize in~\cite[Sec. 4.1.]{Veen2019} the RS mode described in Sec.~\ref{sec:prev-mcf-modes}. However, in these works, the presence of an interferometric matrix (see \eg~\cite[Eq.~(15)]{Veen2019}) is often implicit, since, conversely to our scheme, no linear combinations of these visibilities are computed. 

The 8-step \emph{phase-shifting interferometry}~\cite{Cai04} calibration technique (see Sec.~\ref{sec:calib}) is used in, \eg astronomical imaging~\cite{Rabien06}, and microscopy~\cite{Mann22}. The estimated complex wavefields implicitly encode transmission matrix of the MCF (see~\cite{Sivankutty18bis}) and also embed some unpredictable imperfections in the MCF configuration.
Compared to previous work~\cite{Guerit2021} where each speckle generated by a random SLM configuration had to be \emph{a priori} recorded, this calibration is made only once before any acquisition. 

\subsection{Contributions} \label{sec:contributions}

We provide several contributions to the modeling, understanding and efficiency of MCF-LI imaging.

\noindent
\textbf{Interferometric analysis of MCF imaging}:  Incorporating the physics of wave propagation, we leverage a speckle illumination model to highlight the interferometric nature of the MCF device. The resulting approach, which is new compared to~\cite{Guerit2021}, corresponds to applying symmetric rank-one projections, or SROP\footnote{The ROP terminology was introduced when~\cite{cai2015rop} extended phase retrieval applications~\cite{chen2015exact, Cand2011} to the recovery of a low---(but not necessarily 
  one)---rank matrix via rank-one projections.}, to an interferometric matrix encoding the spectral content of the image. The SROP being controlled by the SLM, we can model the speckles in the sample plane from the random complex amplitude of each core, without assuming Gaussian distributed speckle illuminations~\cite{Guerit2021}. 

\noindent \textbf{Rank-one projected partial Fourier sampling}: Following CS theory, we study a novel sensing model, the rank-one projected partial Fourier sampling of a sparse image (see also Fig.~\ref{fig:comput}), providing a simplified framework for MCF-LI. The theoretical analysis provided in Sec.~\ref{sec:image-reconstruction} has thus an independent interest, the combination of SROPs with a partial Fourier sampling scheme having not been considered previously in the literature. Specifically, from a set of simplifying assumptions, we show that, with high probability (\whp), one can robustly estimate a $K$-sparse image provided that the number of SROPs $M$ and the number of core pairs $Q(Q-1)$ are large compared to $\cl O(K)$ (up to log factors). Our analysis relies on showing that, \whp, the sensing operator satisfies (a variant of) the restricted isometry property which enables us to estimate a sparse image with (a variant of) the basis pursuit denoise program. The number of measurements $M$ is thus reduced compared to recovering first the interferometric matrix with $\cl O(KQ)$ measurements, before estimating the object from this matrix (see Sec.~\ref{sec:intefero-reconstruction}). 

\noindent \textbf{Efficient SROP debiasing}: As explained in~\cite{chen2015exact}, SROP sensing must be \emph{debiased} to allow for signal estimation. This is usually done by doubling the number SROP measurements and computing non-adjacent consecutive SROP differences. By considering \emph{sketching} vectors with unit modulus (but random phases), we propose a more efficient debiasing that simply centers the measurements without doubling their number; a definite advantage when recording experimental measurements.

\noindent \textbf{Simplified calibration}: In Sec.~\ref{sec: exp}, we propose a single-step calibration procedure encompassing most sensing imperfections in a real setup, at the exception of intercore interferences. This calibration, which only requires registering the cores locations and imaging depth, enables us to predict speckle illumination from the programmed SLM configuration (reaching $97\%$ of normalized cross-correlation). Compared to~\cite{Guerit2021} that needed to prerecord the $M$ generated speckle patterns for imaging, we need $\cl O(Q)$ observations for a $Q$ core MCF to accurately model of these $M$ patterns.

\textbf{Notations and conventions:} Light symbols denote scalars (or scalar functions), and bold symbols refer to vectors and matrices (\eg $\eta \in \Rbb$, $g \in L_2(\Rbb)$, $\bs f \in \Rbb^N$, $\bs G \in \Cbb^{N 
\times N}$). We write $\im=\sqrt{-1}$; the identity operator (or $n \times n$ matrix) is $\Id$ (resp. $\Id_n$); the set of $Q \times Q$ Hermitian matrices in $\bb C^{Q \times Q}$ is denoted by $\cl H^Q$; the set of index components is $[M] := \{1,\ldots,M\}$; $\{s_q\}_{q=1}^Q$ is the set $\{s_1, \ldots, s_Q\}$, and $(a_q)_{q=1}^Q$ the vector $(a_1,\ldots,a_Q)^\top$. The notations $\transpose{\cdot}$, 
$\conj{\cdot}$, 
$\tr$,  
$\scp{\cdot}{\cdot}$, 
correspond to the transpose, conjuguate transpose, trace,  
and inner product. 
The $p$-norm (or $\ell_p$-norm) is $\norm{\bs x}{p}:=  (\sum_{i=1}^n |x_i|^p)^{1/p}$ for $\bs x \in \bb C^n$ and $p\geq 1$, with ${\|\cdot\|}={\|\cdot\|_2}$, and $\norm{\bs x}{\infty} := \max_i |x_i|$. Given $\bs A \in \bb C^{n\times n}$, $\bs a \in \bb C^n$ and $\cl S \subset [n]$, the matrix $\bs A_{\cl S}$ is made of the columns of $\bs A$ indexed in $\cl S$, the operator $\diag(\bs A) \in \bb C^n$ extracts the diagonal of $\bs A$,  $\diag(\bs a) \in \bb C^{n\times n}$ is the diagonal matrix such that $\diag(\bs a)_{ii} = a_i$, $\dmat{\bs A}=\diag(\diag(\bs A))$ zeros out all off-diagonal entries of $\bs A$, and $\|\bs A\|$ and $\|\bs A\|_*$ are the operator and nuclear norms of $\bs A$, respectively.  The direct and inverse continuous Fourier transforms in $d$ dimensions (with $d \in \{1,2\}$) are defined by $\hat g(\bs \chi) := \cl F[g](\bs \chi) := \int_{\bb R^d} g(\bs s) e^{-\im 2\pi \bs \chi^\top \bs s} \ud \bs s$, with $g: \bb R^d \to \bb C^d$, $\bs\chi \in \bb R^d$, and $g[\bs s] = \cl F^{-1}[\hat g](\bs s) = 
\int_{\bb R^d} \hat{g}(\bs \chi) e^{ \im 2\pi \bs \chi^\top \bs s} \ud \bs \chi$, with the scalar product $\bs \chi^\top \bs s$ reducing to $k s$ in one dimension. 

\section{MCF lensless imaging} \label{sec: model}

We here develop the sensing model associated with an MCF lensless imager (MCF-LI) in the same speckle imaging conditions as provided in~\cite{Guerit2021}, \ie interfering coherent light components output by the cores of the MCF with random relative delays. As illustrated in Fig.~\ref{fig:LIMCF}(top), an MCF-LI consists of four main parts: a wavefront shaper (SLM), optics, an MCF and a single photo-detector. The SLM shapes the phase of the light that is injected into the cores. The optics include mirrors and lenses used to focus the light into the center of each core, hence preventing multimodal effects. 

As explained below, under a common far-field assumption, MCF-LI can be described as a two-component sensing system applying SROP of a specific interferometric matrix. We show how this model subsumes previous descriptions of the MCF-LI, and end this section highlighting that the SROP and interferometric nature of the model hold beyond the far-field assumptions.  

\subsection{Sensing model description}
\label{sec:model-descr}

An MCF with diameter $D$ contains $Q$ fiber cores with the same diameter $d<D$ (see Fig.~\ref{fig:LIMCF}(c)). Our goal is to observe an object (or sample) which, for simplicity, is planar and defined in a plane $\cl Z$. This plane is parallel to the plane $\cl Z_0$ containing the \emph{distal end} of the MCF, and at distance $z$ from it. For convenience, we assume that the origins of $\cl Z_0$ and $\cl Z$ are aligned, \ie they only differ by a translation in the plane normal direction. In $\cl Z_0$, the $Q$ cores locations are encoded in the set $\Omega := \{\bs p_q\}_{q=1}^Q \subset \Rbb^2$.

As illustrated in Fig.~\ref{fig:LIMCF} (and detailed in Sec.~\ref{sec:setup} and~\cite{Sivankutty2016}), in MCF-LI the laser light wavefront entering the MCF is shaped with a spatial light modulator (SLM) so that both the light intensity and phase can be individually adjusted for each core at the MCF distal end. Mathematically, assuming a perfectly calibrated system, this amounts to setting the $Q$ complex amplitudes $\bs \alpha = (\alpha_1, \ldots, \alpha_Q)^\top \in \Cbb^Q$, coined \emph{sketching vector}, of the electromagnetic field at each fiber core $\bs p_q$ with $q \in [Q]$.

Under the far-field approximation, that is if $z \gg D^2/\lambda$ with $\lambda$ the laser wavelength, the illumination intensity $S(\bs x; \bs \alpha)$ produced by the MCF on a point $\bs x \in \bb R^2$ of the plane $\cl Z$ reads~\cite{Guerit2021} 
\begin{equation} \label{eq:speckle}
    \ts S(\bs x; \bs \alpha) \approx  w(\bs x)\, 
    \big|\!\sum_{q=1}^Q \alpha_q e^{\frac{\im 2\pi}{\lambda z} 
      \bs p_q^\top \bs x}\big|^2,\ w(\bs x) := \tfrac{| \hat{E}_0( \tfrac{\bs x}{\lambda z})|}{(\lambda z)^2}. 
\end{equation}
The window $w(\bs x)$, which relates to the output wavefield $E_0$ of one single core in plane $\cl Z_0$, is a smooth vignetting function defining the imaging field-of-view. Assuming $E_0$ shaped as a Gaussian kernel of width $d$, the FOV width scales like $\frac{\lambda z}{d}$.

\begin{figure}[t]
  \centering
  \subfloat[][]{\includegraphics[width=.65\linewidth]{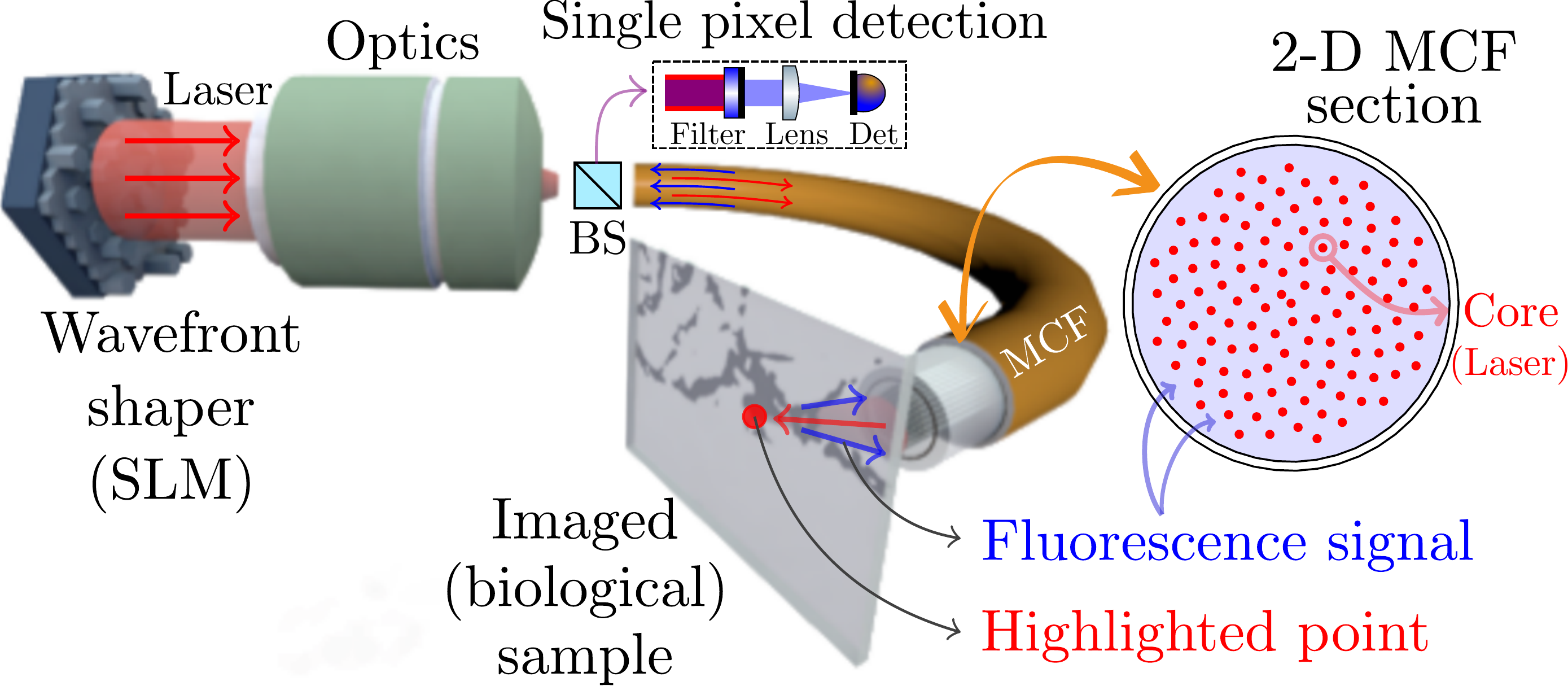}\label{fig:LIMCF-a}}
  \subfloat[][]{\includegraphics[width=.3\linewidth]{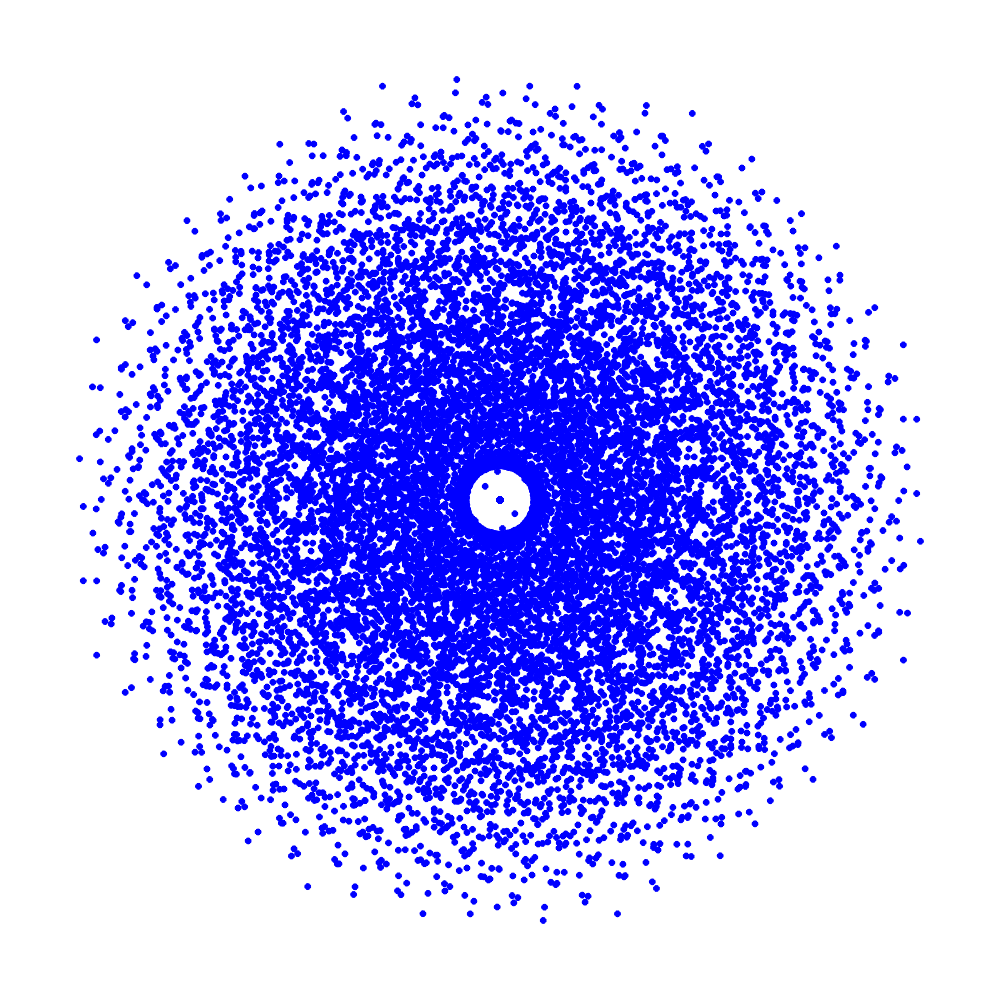}\label{fig:LIMCF-b}}
  \\[5mm]
   \subfloat[][]{\includegraphics[width=.95\linewidth]{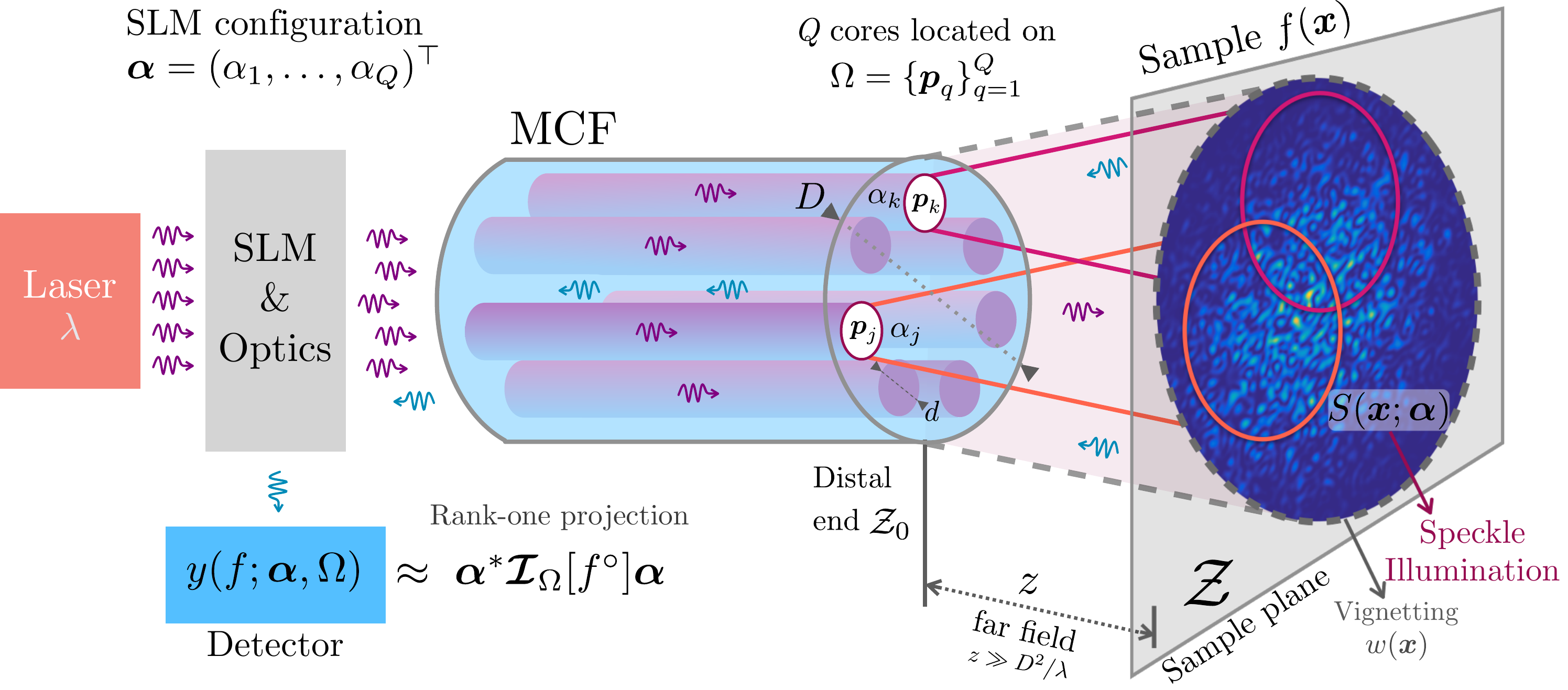}\label{fig:LIMCF-c}}
  \caption{(a) Working principle of MCF-LI with cores arranged 
  in Fermat's golden spiral when the SLM is programmed in raster 
  scanning mode (BS = Beam Splitter). (b) Fourier sampling $\cl V$ corresponding to the core arrangement in (a). (c) Interferometric LI and its link with SROP of the 
  interferometric matrix.}
  \label{fig:LIMCF}
\end{figure}

The sensing model of MCF-LI is established by the following key element: in its endoscope configuration, the sample is observed from the light it re-emits (by fluorescence) from its illumination by $S$, and for each configuration of $S$ a single pixel detector measures the fraction of that light that propagates backward in the MCF (see Fig.~\ref{fig:LIMCF}(a)). Therefore, given the sample fluorophore density map $f(\bs x)$, assuming a short time exposure and low intensity illumination, fluorescence theory tells us that the number of collected photons $y \in \Rbb_+$ follows a Poisson distribution $\cl P(\bar y)$ with average intensity~\cite{Guerit2021}
\begin{align} 
  \ts \bar y&\ts = c \, \scp{S(\cdot, \bs \alpha)}{f} = c \int_{\bb R^2} S(\bs x; \bs \alpha) f(\bs x) \, \ud \bs x \label{eq:speck_proj} \\
&\ts = c \sum_{j,k=1}^Q \alpha_j^* \alpha_k \ \int_{\bb R^2} e^{\frac{\im2\pi}{\lambda z} (\bs p_k - \bs p_j)^\top \bs x}\, \fvign(\bs x) \, \ud \bs x, \notag
\end{align}
where $0<c<1$ represents the fraction of light collected by the pixel detector, and $\fvign := w f$ is the vignetted image, \ie the restriction of $f$ to the domain of the vignetting $w$.

Therefore, assuming $c=1$ for simplicity, if one collects observations $\bs y=(y_1,\ldots, y_M)^\top$, such that $y_m = y(f; \bs \alpha_m, \Omega)$ with distinct vectors $\bs \alpha_m$ ($m \in [M]$), we can compactly write 
\begin{equation} \label{eq:single-ROP-LE}
  \bar y_m = \bs \alpha_m^* \, \intMOm{\fvign} \, \bs \alpha_m = \scp[\big]{\bs \alpha_m\bs \alpha_m^*}{\intMOm{\fvign}}_F,
\end{equation}
where $\scp{\bs A}{\bs B}_F = \tr \bs A^* \bs B$ is the matrix (Frobenius) scalar product between two matrices $\bs A$ and $\bs B$. This amounts to collecting $M$ \emph{sketches} of the Hermitian \emph{interferometric} matrix $\intMOm{\fvign} \in \cl H^Q$, with entries defined by
\begin{equation} \label{eq:intmat}
  \ts  (\intMOm{g})_{jk} := \hat g[\frac{\bs p_j - \bs p_k}{\lambda z}] = \int_{\bb R^2} e^{\frac{\im 2\pi}{\lambda z} 
    (\bs p_k - \bs p_j)^\top \bs x} g(\bs x) \ud \bs x,
\end{equation}
for any function $g: \bb R^2 \to \Rbb$. Under a high photon counting regime, and gathering all possible noise sources in a single additive, zero-mean noise $\bs n$, the measurement model reads 
\begin{equation} \label{eq:ROP-LE}
    \bs y = \ropA \circ \intMOm{\fvign}\ +\ \bs n,
\end{equation}
where the \emph{sketching operator} $\ropA$ defines $M$ SROP~\cite{chen2015exact,cai2015rop} of any Hermitian matrix $\bs H$ with
\begin{equation} \label{eq:SROP_op}
   \ropA(\bs H)  := (\fro{\bs \alpha_m \bs \alpha_m^*}{\bs H} )_{m=1}^M.
\end{equation}
The resulting sensing model is summarized in~Fig.~\ref{fig:comput}.

\begin{figure}[t]  
  \centering
  \includegraphics[width=\linewidth]{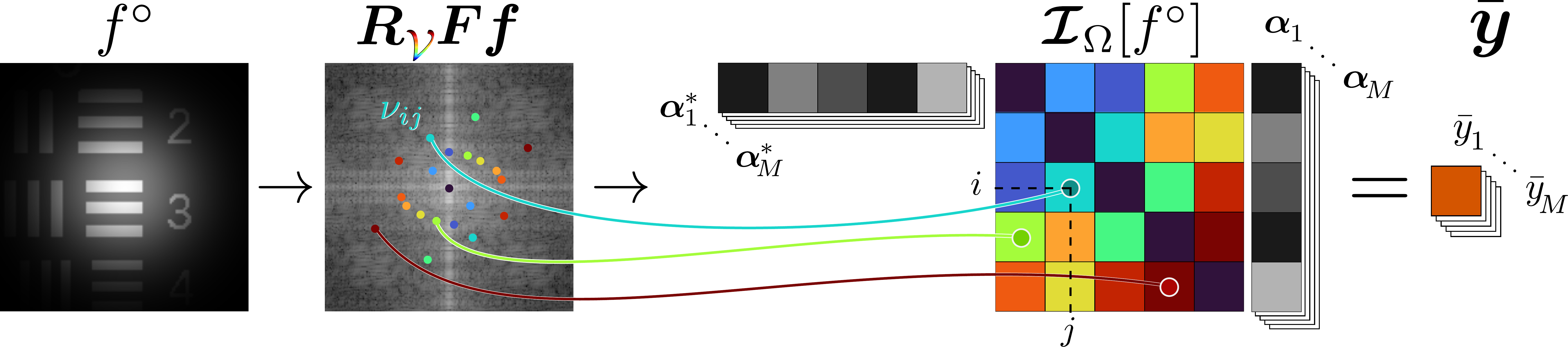}
  \caption{Representation of the sensing model (\ref{eq:single-ROP-LE}). The Fourier transform of the vignetted signal $f^\circ := wf$ is first sampled on the frequencies of the difference set $\cl V := \tinv{\lambda z} (\Omega-\Omega)$. This Fourier sampling is illustrated by restricting the FFT $\bs F$ of the discretized vignetted signal $\bs f$ with $\bs R_{\cl V}$. Next, these samples are shaped into an (hermitian) interferometric matrix $\intMOm{f^\circ} \in \cl H^Q$. Finally,  $M$ SROPs of this matrix are collected from $\bar{\bs y} = (\bar y_m)_{m=1}^M := (\bs\alpha_m^* \intMOm{f^\circ}\bs\alpha_m)_{m=1}^M $. }
  \label{fig:comput}
\end{figure}

From~(\ref{eq:ROP-LE}), MCF-LI corresponds to an \emph{interferometric} system that is linear in $\fvign$. Eq.~\eqref{eq:intmat} and~\eqref{eq:SROP_op} show that it is indeed tantamount to first sampling the 2-D Fourier transform of $\fvign$ over frequencies selected in the difference multiset\footnote{The elements of a multiset are not necessarily unique.}, or \emph{visibilities}, 
\begin{equation}
  \label{eq:visibilities-def}
\cl V :=  \tinv{\lambda z}(\Omega - \Omega) = \{ \bs \nu_{jk} := \tfrac{\bs p_j - \bs p_k}{\lambda z} \}_{j,k=1}^Q,
\end{equation}
\ie $(\intMOm{\fvign})_{jk} = \cl F[\fvign](\bs \nu_{jk})$, and next performing $M$ SROP of $\intMOm{\fvign}$ with the rank-one matrices $\bs \alpha_m \bs \alpha_m^*$ as determined by~$\ropA$. 

Interestingly, the model~\eqref{eq:ROP-LE} shows that we cannot access more information about $\fvign$ than what is encoded in the frequencies of~$\cl V$. Moreover, this sensing reminds the model of radio-interferometry by aperture synthesis~\cite{wiaux2009compressed}---each fiber core plays somehow the role of a radio telescope and each entry of $(\intMOm{\fvign})_{jk}$ probing the frequency content of $\fvign$ on the \emph{visibility}\footnote{The word ``visibility'' being actually borrowed from this context.}~$\bs \nu_{jk}$.
 
Assuming we collect enough $M$ SROP observations, we can potentially estimate the interferometric matrix $\intMOm{\fvign}$, which in turn allows us to estimate $\fvign$ if $\cl V$ (with $|\cl V| \leq Q(Q-1)/2$) is dense enough. Actually, the Fermat's golden spiral distribution $\Omega$ of the cores depicted in Fig.~\ref{fig:LIMCF}(a)---initially studied for its beam forming performances in MCF-LI by raster scanning~\cite{Sivankutty2016} (see below)---displays good properties, as shown in Fig.~\ref{fig:LIMCF}(b). For this arrangement, conversely to regular lattice configurations, all (off-diagonal) visibilities are unique, \ie $|\cl V| = Q(Q-1)/2$.

\subsection{Connection to known MCF-LI modes} \label{sec:prev-mcf-modes}

The MCF-LI model subsumes the Raster Scanning (RS) and the speckle illumination (SI) modes introduced in~\cite{Sivankutty2016,Guerit2021}.
\paragraph{Raster scanning mode} In the RS mode, the light wavefront is shaped (or \emph{beamformed}) with the SLM to focus the illumination pattern on the sample plane, while galvanometric mirrors translate the focused beam by phase shifting, hence ensuring the final imaging of the sample by raster scanning the sample and collecting light at each beamed position. A beamformed illumination is equivalent to set $\bs \alpha = \bs 1 = (1, \ldots, 1)^\top$ in~(\ref{eq:speckle}).

\noindent In this case, the illumination intensity $S$ corresponds to
\begin{equation}
  \label{eq:beam-formed-PSF}
    \ts S(\bs x; \bs 1) \approx  w(\bs x)\, 
    \big|\!\sum_{q=1}^Q e^{\frac{\im 2\pi}{\lambda z} 
    \bs p_q^\top \bs x}\big|^2 = w(\bs x) \big|\cl F[\phi_\Omega(\frac{\bs x}{\lambda z})]\big|^2,
\end{equation}
where $\phi_\Omega$ is the \emph{array factor} of the core arrangement $\Omega$, with, for any finite set $\cl S \subset \bb R^2$,
$\ts \phi_{\cl S}(\bs p) := \sum_{\bs p' \in \cl S}\delta(\bs p - \bs p')$. Expanding~(\ref{eq:beam-formed-PSF}), we also note that $|\cl F[\phi_\Omega(\frac{\bs x}{\lambda z})]\big|^2 = \cl F[\phi_{\cl V}](\bs x)$. 

Arranging the core locations as a discretized Fermat's spiral was shown to focus the beam intensity on a narrow spot whose width scales like $\frac{\lambda z}{D}$~\cite{Sivankutty2016}. This is induced by the constructive interferences in~\eqref{eq:beam-formed-PSF} around ${\bs x \approx \bs 0}$---other locations being associated with almost destructive interferences. 

The two galvanometric mirrors adapt the light optical path of the beam according to a tilt vector $\bs \theta \in \bb R^2$~\cite{Guerit2021}, \ie $\bs \alpha$ is set to $\bs \gamma_{\bs \theta} := \big(\exp(- \frac{\im 2\pi}{\lambda z} \bs \theta^\top \bs p_q)\big)_{q=1}^Q$ and~(\ref{eq:speckle}) provides
\begin{equation*} \label{eq:transl-beam-formed}
  \ts S(\bs x; \bs \gamma_{\bs \theta}) \approx w(\bs x) \cl T_{\bs \theta} \varphi(\bs x),\ \cl T_{\bs \theta} \varphi(\bs x) := \varphi(\bs x - \bs \theta),
\end{equation*}
\ie $\varphi := \cl F[\phi_{\cl V}](\bs x)$ is translated by $\bs \theta$. We can also write, from the symmetry of $\varphi$, 
\begin{equation} \label{eq:conv-model-rs}
  \ts \bar y_{\bs \theta}  = \scp{S(\cdot, \bs \gamma_{\bs \theta})}{f} =\scp{\cl T_{\bs \theta} \varphi}{\fvign} = (\varphi \ast \fvign)(\bs \theta),
\end{equation}
with $\ast$ the 2-D convolution. Therefore, by defining a raster scanning path $\Theta \subset \bb R^2$ for $\bs\theta$ sequentially visiting all positions in a given 2-D domain within a certain resolution, we see that by collecting all RS observations we image a blurred version (by $\varphi$) of $\fvign$ over $\Theta$. The RS mode is thus characterized by the sketching vectors $\bs \alpha \in \{\bs \gamma_{\bs \theta}: \bs \theta \in \Theta \}$.

Moreover, by considering the model~(\ref{eq:single-ROP-LE}) and the multiset $\cl V_0 :=  \{\bs \nu_{jk}: j,k \in [Q], j\neq k\}$ that removes the $Q$ occurrences of the zero frequency from~$\cl V$, for $\bs \theta = \bs 0$, 
$$ 
\ts \bar y_{\bs 0} = \bs 1^\top \intMOm{}[\fvign] \bs 1 = \sum_{\bs \nu \in \cl V} \hat \fvign [\bs \nu] = Q \hat\fvign[\bs 0] + \sum_{\bs \nu \in \cl V_0} \hat \fvign [\bs \nu].
$$
This shows that $\bar y_{\bs 0}$ probes the content of $\fvign$ around the origin if the multiset $\cl V_0$ is dense enough over the support of $\hat \fvign$ with distinct frequencies; in this case $y_{\bs 0} - Q \hat\fvign[\bs 0] = \sum_{\bs \nu \in \cl V_0} \hat \fvign [\bs \nu] \approx c \fvign(\bs 0)$, for some $c>0$. In this context, the narrowness of the focus relates to the density of $\cl V_0$. Moreover,~(\ref{eq:single-ROP-LE}) and~(\ref{eq:intmat}) provide
$$ 
\ts \bar y_{\bs \theta} = \bs \gamma_{\bs \theta}^* \intMOm{}[\fvign] \bs \gamma_{\bs \theta} = \bs 1^\top \intMOm{}[\cl T_{-\bs \theta} \fvign] \bs 1
$$
for any tilt $\bs \theta$, meeting the convolution interpretation in~(\ref{eq:conv-model-rs}).

Despite its conceptual simplicity, the RS mode has a few drawbacks~\cite{Guerit2021}: \emph{(i)} it requires as many illuminations as the target image resolution; \emph{(ii)} due to limited MCF diameter and the chosen core arrangement, the related convolution kernel $\varphi$ is actually spatially varying, which limits the validity of~(\ref{eq:conv-model-rs}).

\paragraph{Speckle Illumination mode}
\label{sec:speckle-illum-mode}

In the SI mode, the sample $f$ is illuminated with random light patterns called \emph{speckles}. These are generated with random core complex amplitudes $\bs \alpha$. Conversely to the RS mode, by recording all speckles illuminations at calibration, SI does not require to know the MCF transfer matrix.

One can interpret SI as a compressive imaging system~\cite{Candes2006, Candes2008, Jacques2010a}. By considering that both $f$ and each illumination intensity $S(\bs x;\bs\alpha)$ are discretized and vectorized as $\bs f \in \Rbb^N$ and $\bs s \in \Rbb^N$, respectively, and gathering in a matrix $\bs S := (\bs s_1,\, \ldots,\bs s_M) \in \Rbb^{N\times M}$ the $M$ discretized speckles obtained from the sketching vectors $\{\bs \alpha_m\}_{m=1}^M$, the model~(\ref{eq:ROP-LE}) becomes 
\begin{equation}
  \bar y_m = \bs s_m^\top \bs f,\ m \in [M],\ \text{or}\ \bs{\bar y} = \bs S^\top \bs f \in \bb R^M_+. \label{eq:SI-raw-sensing-model}  
\end{equation}
If $M$ is adjusted to the sparsity level of $\bs f$ (with $M < N$), the recovery of $\bs f$ from $\bs y$ becomes a classical compressive sensing (CS) problem with the sensing matrix $\bs S$. 

To characterize the properties of the sensing model~\eqref{eq:SI-raw-sensing-model} in this CS framework, the authors in~\cite{Guerit2021} propose to first to center (or \emph{debiase}) this model by computing $\bs y^{\rm c} = \bs y - y^{\rm a} \bs 1_M$ with the measurement average $y^{\rm a}:=\tinv M \sum_{j=1}^M y_j$ (we reinterpret this operation in Sec.~\ref{sec:debiasing}). This provides, from~\eqref{eq:SI-raw-sensing-model}, 
the model
\begin{equation}
  \label{eq:guerit1}
  \bs y^{\rm c} = \sqrt M \bs \Phi \bs{\bar S f} + \bs n^{\rm c},
\end{equation}
with a centered noise $\bs n^{\rm c}:=\bs n- (\tinv M \sum_{j=1}^M n_j)\bs 1_M$,  $\sqrt M \bs \Phi:= \bs D \bs S^\top \bar{\bs S}^{-1}$ and the debiasing matrix $\bs D := (\Id_M - \tinv{M} \bs 1_M \bs 1_M^\top)$, $\bs{\bar S}:=\diag (\bs{\bar s}) \in \Rbb^{N\times N}$, and $\bs{\bar s}:=\bb E_{\bs\alpha} \bs s$. The map $\bs{\bar S} \bs f$ relates to the discretization of the vignetted image $f^\circ$ defined above.

The debiasing above allowed the authors of~\cite{Guerit2021} to hypothesize that $\bs \Phi$ satisfies the Restricted Isometry Property (RIP), a crucial property in the classical CS problem~\cite{Candes2006} ensuring the success of recovery procedures such as the basis pursuit denoise program (see Sec.~\ref{sec:image-reconstruction}).
SI both improves the quality of the reconstructed images and reduces the acquisition time compared to RS. However, the RIP of the related sensing matrix which relies on specific random speckle configurations has not been established, keeping the sample complexity unknown for stable and robust image recovery
Moreover, in SI mode, we must prerecord $M$---object free---illumination speckles to build $\bs \Phi$, before observing the sample in the imaging plane with the same speckles.

\subsection{Generalized MCF-LI sensing} \label{sec:beyond-FF}

We can extend the MCF-LI model~(\ref{eq:single-ROP-LE}) beyond the far-field and identical core diameter assumptions by replacing the interferometric matrix $\intMOm{\fvign}$ with a more general matrix function $\bs G[f]$. 

Given the wavefield $E_{q}(\bs x)$ of the $q$-th core of the MCF in the plane $\cl Z$, the illumination reads
\begin{equation} \label{eq:speckle-gen}
    \ts S(\bs x; \bs \alpha) :=  \big|\!\sum_{q=1}^Q \alpha_q E_{q}(\bs x) \big|^2,
\end{equation}
and similar developments to Sec.~\ref{sec:model-descr} provide
\begin{equation}
  \label{eq:single-ROP-LE-general}
  \bar y_m = \bs \alpha_m^* \, \bs G[f] \, \bs \alpha_m = \scp[\big]{\bs \alpha_m\bs \alpha_m^*}{\bs G[f]}_F
\end{equation}
where we defined, for any function $h: \bb R^2 \to \bb R$, the Hermitian matrix $\bs G[h] \in \cl H^Q$ with entries
\begin{equation} \label{eq:intmat-gen}
  \ts  G_{jk}[h] := \int_{\bb R^2}  E_{j}^*(\bs x) E_{k}(\bs x)  h(\bs x) \ud \bs x.
\end{equation} 
By recording a spatial discretization of the fields $\{E_q(\bs x)\}_{q=1}^Q$, we can thus estimate the forward model~\eqref{eq:intmat-gen}---and thus $\bc H[h] := (\bs \alpha_m^* \, \bs G[h] \, \bs \alpha_m)_{m=1}^M$---for any function $h$, as imposed to solve the inverse problem~\eqref{eq:single-ROP-LE-general} with practical algorithms. 
While slower than the computation of $\intM{\Omega}{\fvign}$ (\eg with a FFT boosting) and its $M$ SROP, estimating $\bc H$ directly integrates many deviations to the interferometric model, with a calibration limited to the observation of $\cl O(Q)$ discretized spatial intensities aimed to yield $\{E_q(\bs x)\}_{q=1}^Q$. We detail in Sec.~\ref{sec:calib} how to practically achieve this calibration.   

\section{Interferometric matrix reconstruction}
\label{sec:intefero-reconstruction}

It is possible to recover the interferometric matrix $\intMOm{}$ from its SROPs $\bs y$, which allows for subsequent image estimation from that matrix. However, as made clear below, that procedure provides a suboptimal sample complexity compared to the direct sensing approach (combining SROPs and the interferometric sensing) proposed in Sec.~\ref{sec:image-reconstruction}.

With Prop.~\ref{prop:pair2pair_unit} provided in Appendix~\ref{app:nyqu-interfer-reconstr}, we first show that $\cl O(Q^2)$ deterministic sketching vectors suffice to reconstruct any interferometric matrix $\intMOm{}$ in a noiseless scenario, \ie a sample complexity upper bound to any further compressive measurements of this matrix.

Recovering $\intMOm{}$ in less than $\cl O(Q^2)$ SROP is possible if this matrix, and thus $f$, respects specific low-complexity models. First, $\bs{\cl I}_0 := \bs{\cl I}_\Omega[\fvign]$ is Hermitian. Moreover, if $\fvign$ is non-negative, this matrix is positive semi-definite since from~(\ref{eq:intmat}), for any $\bs v \in \bb C^Q$,
\begin{equation}
  \label{eq:int_is_psd}
  \begin{gathered}
  \ts \bs v^* \bs{\cl I}_0 \bs v = \int_{\bb R^2} \fvign(\bs x) \sum_{j,k} v_j^* v_k e^{\frac{\im 2\pi}{\lambda z} 
    (\bs p_k - \bs p_j)^\top \bs x} \ud \bs x\\
     \ts = \int_{\bb R^2} \fvign(\bs x) |\bs v^* \bs \rho(\bs x)|^2 \ud \bs x \geq 0,
\end{gathered}
\end{equation}
where $\bs \rho = (\rho_1, \ldots, \rho _Q) \in \bb C^Q$ with $\rho_j(\bs x) := e^{-\frac{\im 2\pi}{\lambda z} \bs p_j^\top \bs x}$.

Second, if $\fvign$ is composed of a few Dirac spikes, \ie if $\fvign (\bs x) = \sum_{i=1}^{K} u_i \delta (\bs x - \bs{x}_i)$ for $K$ coefficients and locations $\{(u_i, \bs x_i)\}_{i=1}^K$, the interferometric matrix has rank-$K$ since~\eqref{eq:intmat} reduces to the sum of $K$ rank-one matrices, \ie
\begin{equation}
\ts \bs{\cl I}_\Omega[\fvign] = \sum_{i=1}^{K} u_i\, \bs \rho(\bs x_i) \bs \rho^*(\bs x_i).\label{eq:rank-k-and-more-interf-model}
\end{equation}

Under this structural assumption, or if $\bs{\cl I}_0$ is well approximated by a rank-$K$ matrix $(\bs{\cl I}_0)_K$, we can recover $\bs{\cl I}_0$ with high probability provided the sketching vectors $\{\bs \alpha_m\}_{m=1}^M$, and thus $\bs{\cl A}$, are random, \ie their entries are \iid from a centered sub-Gaussian distribution~\cite[Thm 1]{chen2015exact}. In particular, with
\begin{equation}
M \geq M_0 = O(K Q),\label{eq:sampcomplex-chen}
\end{equation}
and probability exceeding $1 - \exp(-c M)$, any matrix $\bs{\cl I}_0$ observed through the model
$\bs y = \bs{\cl A}(\bs{\cl I}_0) + \bs \eta$, with bounded noise $\|\bs \eta\|_1 \leq \varepsilon$, can be estimated from 
$$
\tilde{\bs{\cl I}} \in \arg\min_{\bs{\cl I}} \|\bs{\cl I}\|_{\ast}\ \st\ \bs{\cl I} \succcurlyeq 0,\ \|\bs y - \bs{\cl A}(\bs{\cl I})\|_1 \leq \varepsilon.
$$
This solution is instance optimal, \ie for some $C,D>0$,
\begin{equation}
\label{eq:nucl-norm-min}
\ts \|\tilde{\bs{\cl I}} - \bs{\cl I}_0\|_F \leq C \frac{\|\bs{\cl I}_0 - (\bs{\cl I}_0)_K\|_\ast}{\sqrt K} + D \frac{\varepsilon}{M}.
\end{equation}

The sample complexity in~\eqref{eq:sampcomplex-chen} is, however, not optimal since for a $K$-sparse $\fvign$, $\bs{\cl I}_0$ depends only on $\cl O(K)$ parameters in~\eqref{eq:rank-k-and-more-interf-model}. While~\cite{chen2015exact} provides similar results with reduced sample complexity provided $\bs{\cl I}_0$ is, \eg sparse or circulant, these models are not applicable here and we show in Sec.~\ref{sec:image-reconstruction} that a smaller sample complexity is achievable under certain simplifying assumptions. 
 
\section{Image reconstruction}
\label{sec:image-reconstruction}

Let us consider a compressive sensing framework whose objective is to explore the imaging capability of MCF-LI, \ie we study the problem of directly estimating sparse images from their rank-one projected partial Fourier sampling, as driven by the two sensing components in~(\ref{eq:single-ROP-LE}). As proved in Sec.~\ref{sec:rec-anl}, from simplifying assumptions made on both $\fvign$ and the sensing scenario (see Sec.~\ref{sec:work-hyp}), this method achieved reduced sample complexities compared to the approach in Sec.~\ref{sec:intefero-reconstruction}, which are also confirmed numerically in Sec.~\ref{sec:phase-trans-diag}.   
 
\subsection{Working assumptions}
\label{sec:work-hyp}

We first assume a bounded field of view in MCF-LI.
\begin{assumption}[Bounded FOV]
 \label{h:bounded-FOV} 
The support of the vignetting window $w(\bs x)$ in~(\ref{eq:speckle}) is contained in a domain $\cl D := [-L/2, L/2] \times [-L/2, L/2]$ with $L := \frac{c \lambda z}{d}$, for $c>0$ depending on the (spectrum of the) output wavefield $E_0$ in~(\ref{eq:speckle}), and $w=0$ on the frontier of~$D$. 
\end{assumption}
\noindent Therefore, supposing $f$ bounded, we have $\supp \fvign \subset \cl D$ and $\fvign = 0$ over the frontier of $D$. 

We also need to discretize~$\fvign$ by assuming it bandlimited.
\begin{assumption}[Bounded and bandlimited image]
\label{h:band-limitedness-fvign} 
The image $f$ is bounded, and $\fvign$ is bandlimited with bandlimit $\frac{W}{2}$, with $W := \frac{N_1}{L}$ and $N_1 \in \bb N$, \ie $\cl F[\fvign](\bs \chi) = 0$ for all $\bs \chi$ with $\|\bs \chi\|_\infty \geq \frac{W}{2}$.
\end{assumption}
As will be clear below, this assumption enables the computation of the interferometric matrix $\bc I_{\Omega}[\fvign]$ from the discrete Fourier transform of the following discretization of $\fvign$. 

From~(\ref{h:bounded-FOV}) and~(\ref{h:band-limitedness-fvign}) the function $\fvign$ can be identified with a vector $\bs f \in \bb R^N$ of $N=N_1^2$ components. Up to a pixel rearrangement, each component $f_j$ of $\bs f$ is related to one specific pixel of $\fvign$ taken in the $N$-point grid
$$
\ts \cl G_N := \frac{L}{N_1} \{(s_1,s_2)\}_{s_1,s_2=-\frac{N_1}{2}}^{\frac{N_1}{2}-1} \subset \cl D.
$$
The discrete Fourier transform (DFT) $\hat{\bs{f}}$ of $\bs f$ is then computed from the 2-D  DFT matrix $\bs F \in \bb C^{N \times N}$, \ie $\hat{\bs{f}} = \bs F \bs f \in \bb C^N,\quad\bs F := \bs F_{1} \otimes \bs F_{1}$,
with $(\bs F_{1})_{kl} =  e^{-\frac{\im2\pi}{N_1} kl}/\sqrt{N_1}$, $k,l \in [N_1]$, and the Kronecker product $\otimes$. Each component of $\hat{\bs{f}}$ is related to a 2-D frequency~of
$$
\ts \hat{\cl G}_N := \frac{W}{N_1} \{\chi_1, \chi_2\}_{\chi_1,\chi_2=-\frac{N_1}{2}}^{\frac{N_1}{2}-1} \subset [-\frac{W}{2},\frac{W}{2}] \times [-\frac{W}{2},\frac{W}{2}].
$$
 
We need now to simplify our selection of the visibilities.  
\begin{assumption}[Distinct on-grid non-zero visibilities]
  \label{h:distinct-visib}
All non-zero visibilities in  $\cl V_0 = \cl V \setminus \{\bs 0\}$ belong to the regular grid $\hat{\cl G}_N$, \ie $\cl V_0 \subset \hat{\cl G}_N$, and are unique, which means that $|\cl V_0| = Q(Q-1)$.
\end{assumption}  

Anticipating over Sec.~\ref{sec:rec-anl}, assumptions~\ref{h:bounded-FOV} and~\ref{h:band-limitedness-fvign} show that $\intMOm{\fvign}$ can be computed from $\bs F \bs f$; 
for each visibility $\bs \nu_{jk} \in \cl V$, there is an index $\bar l = \bar l(j,k) \in [N]$ such that      
$\ts (\intMOm{\fvign})_{j,k} = \varpi\, (\bs F \bs f)_{\bar l}$, where $\varpi := \frac{L^2}{\sqrt N}$ can be found from the continuous interpolation formula of the Shannon-Nyquist sampling theorem.

Moreover, from~\ref{h:distinct-visib}, $\bar l(j,k)$ is unique for all $j\neq k$ (\ie $\nu_{j,k} \in \cl V_0$), and since $(\intMOm{\fvign})_{j,j} = \varpi (\bs F \bs f)_{0}$ (\ie $\bar l(j,j) = 0$) for all $j \in [Q]$, we get the equivalence 
\begin{equation}
  \label{eq:equiv-cont-discrt-interfero}
  \intMOm{\fvign} = \varpi\,\cl T(\bs F \bs f),
\end{equation}
where the operator $\cl T$ is such that, for all $j,k \in [Q]$ and $\bs u \in \bb C^N$, $(\cl T(\bs u))_{jk}$ equals $u_0$ if $j=k$, and $u_{\bar l(j,k)}$ otherwise.

Consequently, if  $\fvign$ has zero mean, $(\bs F \bs f)_0 = 0$ and
\begin{equation}
  \label{eq:equiv-frob-l2}
\ts \frac{1}{\varpi^2} \|\intMOm{\fvign}\|_F^2 = \|\bs R_{\overline{\cl V}_0} \bs F \bs f\|^2,
\end{equation}
with $\bs R_{\cl S} = \Id_{\cl S}^\top$ the restriction operator defined for any $\cl S \subset [N]$, and $\overline{\cl V}_0 = \{\,\bar l(j,k): j,k \in [Q], j \neq k\} \subset [N]$ the index set of $\bs F \bs f$ related to the off-diagonal entries of $\intMOm{\fvign}$ (with $|\cl V_0|=|\overline{\cl V}_0|$ from~\ref{h:distinct-visib}).

We need to \emph{regularize} the (ill-posed) MCF-LI imaging problem by supposing that $\bs f$ is \emph{sparse} in the canonical basis. 
\begin{assumption}[Sparse sample image]
  \label{h:sparsifying-basis}
  The discrete image $\bs f$ is $K$-sparse (in the canonical basis):
$\bs f \in \Sigma_K := \{\bs v: \|\bs v\|_0 \leq K\}$.
\end{assumption}
\noindent While restrictive, our experiments in Sec.~\ref{sec: exp} show experimentally that other sparsity priors are compatible with our sensing scheme, \eg the TV norm, postponing to a future work a theoretical justification of such possible extensions.

Our theoretical analysis leverages the tools of compressive sensing theory~\cite{Candes2006a,foucart2017mathematical}. In particular, as stated in the next assumption, we require that the interferometric matrix---actually, its non-diagonal entries encoded in the visibilities of $\cl V_0$---captures enough information about any sparse image~$\bs f$. 
\begin{assumption}[RIP$_{\ell_2/\ell_2}$ for visibility sampling]
\label{h:rip-visibility}
Given a sparsity level $K$, a distortion $\delta > 0$, and provided
\begin{equation} \label{eq:samp-complex-type}
|\cl V_0| = Q(Q-1) \geq \delta^{-2} K\, {\rm plog}(N,K, \delta),
\end{equation}
for some polynomials ${\rm plog}(N,K,1/\delta)$ of $\log N$, $\log K$ and $\log 1/\delta$, the matrix $\bs \Phi := \sqrt{N} \bs R_{\overline{\cl V}_0} \bs F$ respects the $(\ell_2/\ell_2)$-restricted isometry, or RIP$_{\ell_2/\ell_2}(\Sigma_K,\delta)$, \ie
\begin{equation} \label{eq:rip-def}
\ts (1- \delta) \|\bs v\|^2 \leq \tinv{|\cl V_0|}\|\bs \Phi \bs v\|^2_2 \leq (1+ \delta) \|\bs v\|^2,\ \forall \bs v \in \Sigma_K.
\end{equation} 
\end{assumption}
As will be clear later, combined with \eqref{eq:equiv-frob-l2}, this assumption ensures that two different sparse images lead to two distinct interferometric matrices, a key element for stably estimating images from our sensing model (see Prop.~\ref{prop:L2L1}). 

We specify now the distribution of the sketching vectors~$\bs \alpha$.
\begin{assumption}[Random sketches with unit modules]
  \label{h:sketch-distrib}
The sketching vectors $\{\bs \alpha_m\}_{m=1}^M$ involved in~\eqref{eq:ROP-LE} have components \iid as the random variable $\alpha \in \bb C$, with $\bb E \alpha = 0$ and $|\alpha| = 1$. 
\end{assumption}
The sketching vectors are thus sub-Gaussian, since the sub-Gaussian norm $\|\alpha_q\|_{\psi_2} = \| |\alpha_q|\|_{\psi_2} = 1$ is bounded (see~\cite[Sec 5.2.3]{Vershynin2010}). While motivated by the MCF-LI application where the SLM mainly acts on the phase of the core complex amplitudes, this assumption enables a \emph{debiasing trick}, described in Sec.~\ref{sec:debiasing}, which simplifies the theoretical analysis detailed in Sec.~\ref{sec:rec-anl}

\subsection{Rationale and limitations of our assumptions}
\label{sec:limitations}

We here discuss the rationale and limitations of our assumptions. First, both Assumptions~\ref{h:distinct-visib} and~\ref{h:rip-visibility} are built on the multiset $\cl V_0$ (listing all the non-zero visibilities) and not $\cl V$. Anticipating over ~Sec.~\ref{sec:rec-anl}, this choice is imposed by the SROP measurements; they are strongly biased by the diagonal elements of the projected matrices (see also Lemma.~\ref{lem:mean-aniso-srop} in App.~\ref{app:proof-prop-rip-ileop}). As a result, we need a \emph{debiased} sensing model (see Sec.~\ref{sec:debiasing}) removing the diagonal of the interferometric matrix, and hence the zero image frequency.

Second, while we do not prove that the visibility set~$\cl V_0$ defined by the Fermat's spiral core arrangement $\Omega$ in MCF-LI verifies~\ref{h:rip-visibility}, we invoke existing results characterizing compressive sensing with partial Fourier sampling---as established for instance in the context of tomographic and radio interferometric applications~\cite{Candes2008,wiaux2009compressed}---to prove the existence of a visibility set respecting~\ref{h:rip-visibility}. 
For example, from~\cite[Thm 12.31]{foucart2017mathematical}, if $Q(Q-1) \geq C \delta^{-2} K \log^4(N)$ (for some constant $C>0$) and the set of $Q(Q-1)$ visibilities $\ol{\cl V}_0$ are picked uniformly at random in $[N]$, then $\bs \Phi$ respects the RIP$_{\ell_2/\ell_2}(\Sigma_K,\delta)$ with probability exceeding $1 - N^{-\log^3 N}$.

Third, as stated by Assumption~\ref{h:distinct-visib}, our analysis expects that each visibility, except the 0 frequency, is only observed once. However, for large values of $Q$ (and certainly if $Q(Q-1)>N$), low-frequency visibilities tend to occupy the same points in the gridded frequency domain $\hat{\cl G}_N$. For instance, for a 1-D configuration, the pdf of the visibilities would be centered and triangular---with thus increasing multiplicity at low-frequency---if the frequencies were drawn uniformly at random in a given frequency interval. 

Relaxing~\ref{h:distinct-visib} would require adapting our developments in new directions. To account for possible multiplicities in the visibilities, we can indeed introduce a weighting matrix $\bs W$ encoding the number of entries in $\intMOm{\fvign}$ that are related to the same frequency, \ie using the notations defined above, $W_{qq} := |\{(j,k) : \bar l(j,k) = q\}|$, $1\leq q\leq N$. In that case,~\eqref{eq:equiv-frob-l2} becomes 
$$
\ts \frac{1}{\varpi^2} \|\intMOm{\fvign}\|_F^2 = \|\bs W \bs R_{\overline{\cl V}_0} \bs F \bs f\|^2.
$$
However, as our analysis requires that $\|\bs W \bs R_{\overline{\cl V}_0} \bs F \bs f\|^2$ is approximately proportional to the norm of $\|\bs f\|^2$ if $\bs f$ is sparse, \ie the matrix $\bs \Phi' := \bs W \bs R_{\overline{\cl V}_0} \bs F$ should respect the RIP~\cite{Candes2006}, we then need to adapt Assumption~\ref{h:rip-visibility} to that matrix. Unfortunately, as soon as $\bs W \neq \Id$, recovering sparse signals from their random partial Fourier sampling imposes to reweight $\bs \Phi'$ to cancel out the impact of $\bs W$~\cite{krahmer2013stable,adcock_17}, both for ensuring the RIP of this matrix and for the stability of the numerical reconstructions. Unfortunately, as we only access to the SROP of the interferometric matrix, it is unclear how to introduce this cancellation in our sensing model.

Finally, Assumption~\ref{h:sparsifying-basis} is restricted to sparse signals in the canonical basis. Generalizing our recovery guarantees developed in Sec.~\ref{sec:rec-anl} to general sparsifying bases $\bs\Psi \neq \Id$ is, however, challenging. Certain bases, such as the Haar wavelet basis $\bs \Psi$, includes the constant vector. Assuming it placed on the 1-st column of $\bs \Psi$, we easily check that
the RIP$_{\ell_2/\ell_2}$ in~\eqref{eq:rip-def} 
cannot hold anymore for $\bs \Phi = \sqrt N \bs R_{\bar{\cl V}_0} \bs F \bs \Psi$; since $\bar{\cl V}_0$ excludes the DC frequency, taking $\bs v = \lambda \bs e_1 + \bs e_k \in \Sigma_2$ with a sufficiently large value $\lambda$ breaks~\eqref{eq:rip-def} as $\bs\Phi\bs v = \sqrt N \bs R_{\bar{\cl V}_0} \bs F \bs \Psi \bs e_k$.

\subsection{Debiased sensing model}
\label{sec:debiasing}

As made clear in Sec.~\ref{sec:rec-anl}, the estimation of $\bs f$ requires a \emph{debiasing} of the MCF-LI measurements imposed by the properties of the SROP operator $\bs{\cl A}$ in~(\ref{eq:single-ROP-LE}). We follow a debiasing inspired by~\cite{Guerit2021} (and allowed by~\ref{h:sketch-distrib}), with a reduced number of measurements compared to the method proposed in~\cite{chen2015exact}.

From~\eqref{eq:ROP-LE}, we define the \emph{debiased} measurements
\begin{align}
\ts  y^{\rm c}_m := y_m - \tfrac{1}{M} \sum_{j=1}^M y_j = \scp{\bs A^{\rm c}_m}{\bc I}_F + n^{\rm c}_m,\label{eq:deb-z} 
\end{align}
with the centered and the average matrices $\bs A^{\rm c}_m = \bs \alpha_m\bs \alpha_m^* - \bs A^{\rm a}$ and $\bs A^{\rm a} = \frac{1}{M}\sum_{j=1}^M \bs \alpha_j\bs \alpha_j^*$, respectively, $\bc I := \intMOm{\fvign}$, and noise $n^{\rm c}_m := n_m - \tfrac{1}{M} \sum_{j=1}^M n_j$ with $\bb E |n^{\rm c}_m|^2 = (1 - \frac{1}{M}) \bb E |n_m|^2$.

Introducing the debiased sensing operator
\begin{equation}
\bc A^{\rm c}: \bc J \in \cl H^Q\mapsto \big(\scp{\bs A^{\rm c}_m}{\bc J}\big)_{m=1}^M \in \bb R^M,\label{eq:centered-srop-op}
\end{equation}
which respects $\bc A^{\rm c}(\bc J) = \bc A^{\rm c}(\bc J_{\rm h})$ with the \emph{hollow} matrix $\bc J_{\rm h} := \bc J - \bc J_{\rm d}$ (\ie $\diag(\bc J_{\rm h})=\bs 0$) since each $\bs A^{\rm c}_m$ is hollow from~\ref{h:sketch-distrib}, we can compactly write 
\begin{equation} \label{eq:centered-sensing-model}
\ts \bs y^{\rm c} = (y^{\rm c}_1,\,\ldots, y^{\rm c}_M)^\top  = \bc A^{\rm c}(\bc I_{\rm h})+ \bs n^{\rm c}. 
\end{equation}
The debiasing model thus senses, through $\bc I_{\rm h}$, the off-diagonal elements of $\intMOm{\fvign}$. We will show below that the combination of $\bc A^{\rm c}$ with the interferometric sensing respects a variant of the RIP property, thus enabling image reconstruction guarantees.

\begin{figure*}[t]
  \centering   
\scalebox{1.1}{\raisebox{2mm}{\includegraphics[width=.28\textwidth]{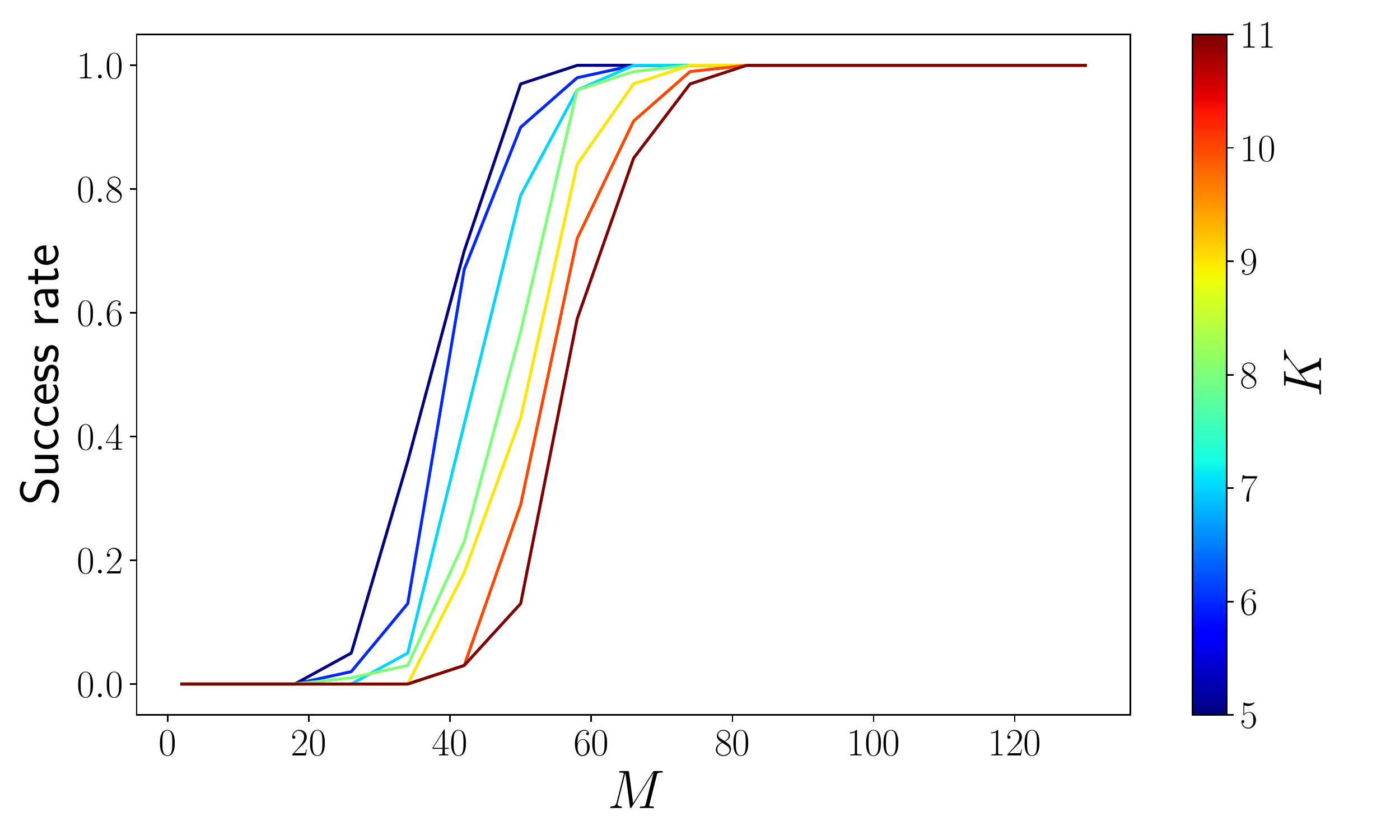}}  \includegraphics[width=.6\textwidth]{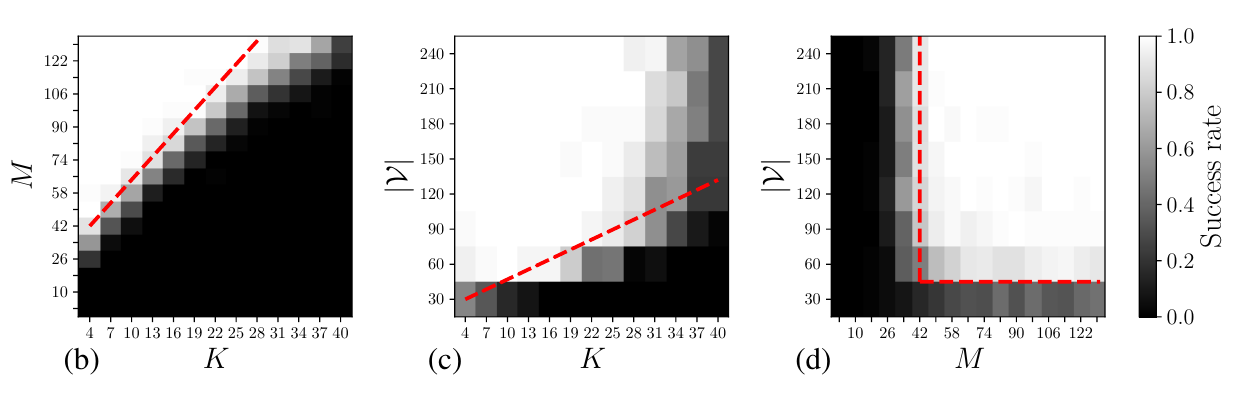}}
  \put(-515,7){\small (a)}
  \caption{(a) Transition curves obtained with $|\cl V|=240$, ensuring widespread Fourier sampling. The success rate is computed from $100$ trials. The transition abscissa shifts to the right for an increasing number $K$ of spikes in $\bs f$, indicating more SROP are necessary to reconstruct the inteferometric matrix. (b-d) Phase transition diagrams showing $M$ SROP of a $Q \times Q$ interferometric matrix for a $K$-sparse object $\bs f$ (with $|\cl V|=240$ in (b), $M=122$ in (c), and $K=4$ in (d)). One considers a uniformly random 1-D core arrangement and SROP using circularly-symmetric unit-norm random $\{ \bs\alpha_m \}_{m=1}^M$. Each pixel is constructed with $80$ reconstruction trials solving~\eqref{eq:lasso} where we consider success if SNR$\ge 40$dB. The probability of success ranges from black (0\%) to white (100\%). Dashed red lines link the transition frontiers to the samples complexities provided in Sec.~\ref{sec:work-hyp} and Sec.~\ref{sec:rec-anl}. In (c), the line only coincides with low values of $\cl V$ due to multiplicity effects.  \label{fig:transition}}
\end{figure*}

\subsection{Reconstruction analysis} \label{sec:rec-anl}

We show now that we can estimate a sparse image $\bs f$ from its sensing~\eqref{eq:centered-sensing-model}. From~\ref{h:bounded-FOV}-\ref{h:sketch-distrib}, it be recast as 
\begin{equation} \label{eq:debiased-ILEop-sensing-model}
\bs y^{\rm c} = \ILEop(\bs f) + \bs n^{\rm c},
\end{equation}
where, from the equivalence~\eqref{eq:equiv-cont-discrt-interfero}, the sensing operator $\ILEop$ reads
\begin{equation}
\ts \ILEop : \bs v \in \bb R^N\ \mapsto\ \varpi \ropA^{\rm c}\big(\cl T (\bs F \bs v)\big) \in\bb R^M_+.\label{eq:ileop-def}
\end{equation}
We propose to estimate $\bs f$ by solving the basis pursuit denoise program with an $\ell_1$-norm fidelity (or BPDN$_{\ell_1}$), \ie
\begin{equation}
  \label{eq:BPDN}
\ts  \tilde{\bs f} = \arg\min_{\bs v \in \bb R^N} \|\bs v\|_1 \st 
  \|\bs y^{\rm c} - \ILEop(\bs v) \|_1 \leq \epsilon, \hspace{-5mm} \tag{BPDN$_{\ell_1}$}
\end{equation}
The specific $\ell_1$-norm fidelity of this program is motivated by the properties of the SROP operator~$\ropA^{\rm c}$, and this imposes us to set $\epsilon \geq \|\bs n^{\rm c}\|_1$ to reach feasibility.  We indeed show below that $\ILEop$, through its dependence on $\ropA^{\rm c}$, respects a variant of the RIP, the RIP$_{\ell_2/\ell_1}(\Sigma_K,{\sf m}_K,{\sf M}_K)$: given a sparsity level $K$, and two constants $0<{\sf m}_K<{\sf M}_K$, this property imposes
\begin{equation}
\ts {\sf m}_K \|\bs v\| \leq \tinv{M} \|\ILEop(\bs v)\|_1 \leq {\sf M}_K \|\bs v\|, \quad \forall \bs v \in  \Sigma_K.\label{eq:RIP-L2L1-def}
\end{equation}
Under this condition, the error $\|\bs f -\tilde{\bs f}\|$ is bounded, \ie instance optimal~\cite{foucart2017mathematical}. This is shown in the following proposition (inspired by~\cite[Lemma 2]{chen2015exact} and proved Appendix~\ref{app:proof_L2L1}). 
\begin{proposition}[$\ell_2/\ell_1$ instance optimality of~\ref{eq:BPDN}]
  \label{prop:L2L1}
  Given $K$, if there exists an integer $K' > 2K$ such that, for $k \in \{K', K + K'\}$, the operator $\ILEop$ has the RIP$_{\ell_2/\ell_1}(\Sigma_k, {\sf m}_{k}, {\sf M}_{k})$ for constants $0 < {\sf m}_{k} < {\sf M}_{k} < \infty$, and if
\begin{equation}
\ts  \tinv{\sqrt 2} {\sf m}_{K+K'} - {\sf M}_{K'} \frac{\sqrt K}{\sqrt{K'}} \geq \gamma >0,\label{eq:RIP-L2L1-bound-condition-BPDN}
\end{equation}
then, for all $\bs f$ sensed through $\bs y^{\rm c} = \ILEop(\bs f) + \bs n^{\rm c}$ with bounded noise $\|\bs n^{\rm c}\|_1 \leq \epsilon$,  the estimate $\tilde{\bs f}$ provided by~\ref{eq:BPDN} satisfies
\begin{equation}
  \label{eq:bpdn-inst-opt}
\ts  \| \bs f - \tilde{\bs f}\| \leq C_0 \frac{\|\bs f -\bs f_K\|_1}{\sqrt{K'}} + D_0 \frac{\epsilon}{M},
\end{equation}
for two values $C_0=\cl O({\sf M}_{K'}/\gamma)$ and $D_0=\cl O(1/\gamma)$.
\end{proposition}
Notice that~\eqref{eq:RIP-L2L1-bound-condition-BPDN} is satisfied if 
\begin{equation}
\ts  K'  > 8 \Big(\frac{{\sf M}_{K'}}{{\sf m}_{K+K'}}\Big)^2 K,\label{eq:bound-condition-bpdn-reloaded}
\end{equation}
in which case $\gamma = \tinv{2\sqrt 2} {\sf m}_{K+K'}$, and, from App.~\ref{app:proof_L2L1},  $C_0 = 2(\sqrt 2+1) ({\sf M}_{K'}/{\sf m}_{K+K'})+2$ and $D_0 =  4(\sqrt 2+1)/{\sf m}_{K+K'}$. 

Interestingly, if both $M$ and $Q(Q-1)$ sufficiently exceed $K$, the operator $\ILEop$ respects the RIP$_{\ell_2/\ell_1}$ with high probability.  
\begin{proposition}[RIP$_{\ell_2/\ell_1}$ for $\ILEop$] \label{prop:rip-ileop}
Assume that assumptions~\ref{h:bounded-FOV}-\ref{h:sketch-distrib} hold, with~\ref{h:rip-visibility} set to sparsity level $\splev >0$ and distortion $\delta = 1/2$ over the set $\Sigma_\splev$. For some values $C, c >0$ and $0<c_\alpha<1$ only depending on the distribution of $\alpha$, if   
\begin{equation} \label{eq:QM_vs_K}
\ts M \geq C \splev \ln(\frac{12 e N}{\splev}),\ Q(Q-1) \geq 4 \splev\, {\rm plog}(N,\splev, \delta),
\end{equation} 
then, with probability exceeding $1 - C \exp(-c M)$, the operator $\ILEop$ respects the RIP$_{\ell_2/\ell_1}(\Sigma_\splev, {\sf m}_\splev, {\sf M}_\splev)$ with
\begin{equation}
\ts m_\splev > \frac{\varpi c_\alpha}{3\sqrt 2}\frac{\sqrt{|\cl V_0|}}{\sqrt N},\ \text{and}\ M_\splev < \frac{8\varpi}{3}\frac{\sqrt{|\cl V_0|}}{\sqrt N}. \label{eq:rip-l2l1-ilerop}
\end{equation}
\end{proposition}
In this proposition, proved in Appendix~\ref{app:proof-prop-rip-ileop}, the constants in~\eqref{eq:rip-l2l1-ilerop} have not been optimized and may not be tight, \eg they do not depend on $\splev$. 

Combining these last two propositions and using the (non-optimal) bounds~(\ref{eq:rip-l2l1-ilerop}) that are independent of $\splev$, since $8 ({\sf M}_{K'}/{\sf m}_{K+K'})^2 < \frac{1024}{c^2_\alpha}$,~\eqref{eq:bound-condition-bpdn-reloaded} holds if $K'> 1024 K/ c^2_\alpha$. Therefore, provided $\ILEop$ satisfies the RIP$_{\ell_2/\ell_1}(\Sigma_\splev,{\sf m}_\splev, {\sf M}_\splev)$ for $\splev\in\{K',K+K'\}$, the instance optimality~\eqref{eq:bpdn-inst-opt} holds with  
$$
\ts C_0 <  \frac{16(\sqrt{2} + 2)}{c_\alpha} = \cl O(1),\quad D_0 = \cl O\Big(\frac{\sqrt{N}}{\varpi \sqrt{|\cl V_0|}}\Big) = \cl O\big(\frac{N}{L^2 Q}\big).
$$
While the constraint on $K'$ imposes a high lower bound on $M$ when the sample complexity~\eqref{eq:QM_vs_K} is set to $\splev=(K+K')> (1024/c^2_\alpha +1) K$---as necessary to reach the RIP {w.h.p.}---the impact of the sparsity error $\|\bs f - \bs f_K\|$ in~\eqref{eq:bpdn-inst-opt} is, however, attenuated by $1/\sqrt{K'} < c_\alpha / (32 \sqrt{K})$. 

For a fixed FOV $L^2$, we also observe a meaningful amplification of the noise error by $D_0$ when the sampling grid $\cl G_N$ is too large compared to $Q$: if the number of pixels $N$ is too small,~\ref{h:band-limitedness-fvign} may not be verified, since the image bandwidth lower bounds $N$; if $N$ is too large the noise error in~(\ref{eq:bpdn-inst-opt}) is vacuous. 

\subsection{Phase transition diagrams}
\label{sec:phase-trans-diag}

We now compare our recovery guarantees with empirical reconstructions obtained on extensive Monte Carlo simulations with $S$ trials and varying parameters $K$, $Q$ and $M$.

To save computations, we adopt a simplified setting where~\eqref{eq:debiased-ILEop-sensing-model} is adapted to the sensing of 1-D zero mean sparse vectors in $\bb R^{N=256}$, without any vignetting, \ie $\fvign=f$, and 1-D MCF core locations. At each simulation trial with fixed $(K,Q,M)$, we verified~\ref{h:bounded-FOV}--\ref{h:sparsifying-basis} by picking the 1-D cores locations $\{p_q\}_{q=1}^{Q} \subset \bb R$ uniformly at random without replacement in $\big[\frac{-N}{2},\frac{N}{2}\big]$, and $M$ sketching vectors $\{\bs \alpha_m\}_{m=1}^M$ \iid as $\bs \alpha \in \bb C^Q$ with $\alpha_k \sim_{\iid} = e^{\im \cl U([0,2\pi[)}$, $k \in [Q]$. A zero average vector $\bs f \in \bb R^{N=256}$ was randomly generated with a $K$ sparse support picked uniformly at random in $[N]$, its $K$ non-zero components obtained with $K$ \iid Gaussian values $\cl N(0,1)$ to which we subtracted their average. The interferometric matrix was computed from~(\ref{eq:equiv-cont-discrt-interfero}) (with $L=\lambda=z=1$) using the 1-D FFT matrix $\bs F_1$. 

To reconstruct $\bs f$, we solved the \emph{Lasso} program
\footnote{We used SPGL1~\cite{pareto} (Python module: \url{https://github.com/drrelyea/spgl1}).}~\cite{pareto}, 
\begin{equation} \label{eq:lasso}
\tilde{\bs f} = \argmin_{\bs v} \tfrac{1}{2}\| \bs y^{\rm c} - \ILEop (\bs v)\|^2~\st~\norm{\bs v}{1}\leq \tau
\end{equation}
with $\tau=\| \bs f \|_1$ set to the actual $\ell_1$-norm of the discrete object. Eq.~\eqref{eq:lasso} is equivalent to~\ref{eq:BPDN} in a noiseless setting (\ie $\epsilon = 0$) as it includes an equality constraint $\bs y^{\rm c} = \ILEop\bs (\bs f)$~\cite[Prop. 3.2]{foucart2017mathematical}. In the sparse and noiseless sensing scenario set above, we thus expect from~\eqref{eq:bpdn-inst-opt} in Prop.~\ref{prop:L2L1} that $\tilde{\bs f} = \bs f$ if $\bc B$ is RIP$_{\ell_2/\ell_1}$, \ie if both $M$ and $Q(Q-1)$ sufficiently exceeds $K$ from Prop.~\ref{prop:rip-ileop}. 

In Fig.~\ref{fig:transition}(a) and Fig.~\ref{fig:transition}(b-d), the success rates---\ie the percentage of trials where the reconstruction SNR exceeded 40 dB---were computed for $S$ set to $80$ and $100$ trials per value of $(K,Q,M)$, respectively, and for a range of $(K,Q,M)$ specified in the axes. Since~\ref{h:distinct-visib} was partially verified, we tested these rates in function of the averaged value of $|\cl V|\leq Q(Q-1)$ (which had a std $\leq 0.08N$) over the $S$ trials instead of $Q(Q-1)$.
We observe in Fig.~\ref{fig:transition}(b) that high reconstruction success is reached as soon as $M\geq C K$, with $C\simeq11$, in accordance with~\eqref{eq:QM_vs_K} in Prop.~\ref{prop:rip-ileop} (up to log factors). Fig.~\ref{fig:transition}(c) highlights that the Fourier sampling $|\cl V|$ (and thus $Q$) must increase with $K$. At small value of $Q$, we reach high reconstruction success if $|\cl V| \approx Q(Q-1) \geq C' K$, with $C' \simeq 10$, in agreement with~\eqref{eq:QM_vs_K} (up to log factors). However, as $Q$ rises that linear trend is biased since the multiplicities in $\cl V$ increases, \ie $Q(Q-1) - |\cl V| \geq 0$ increases. As expected from~\eqref{eq:QM_vs_K}, the transition diagram in Fig.~\ref{fig:transition}(d) shows that at a fixed $K=4$, both $M$ and $|\cl V|$ must reach a threshold value to trigger high reconstruction success. 
In Fig.~\ref{fig:transition}(a), which displays several transition curves of the success rate vs. $M$ for different values of $K$ at $|\cl V|=240$, the failure-success transition is shifted towards an increasing number of SROPs when $K$ increases.

\begin{figure*}[t]
  \centering
  \includegraphics[height=.15\textwidth]{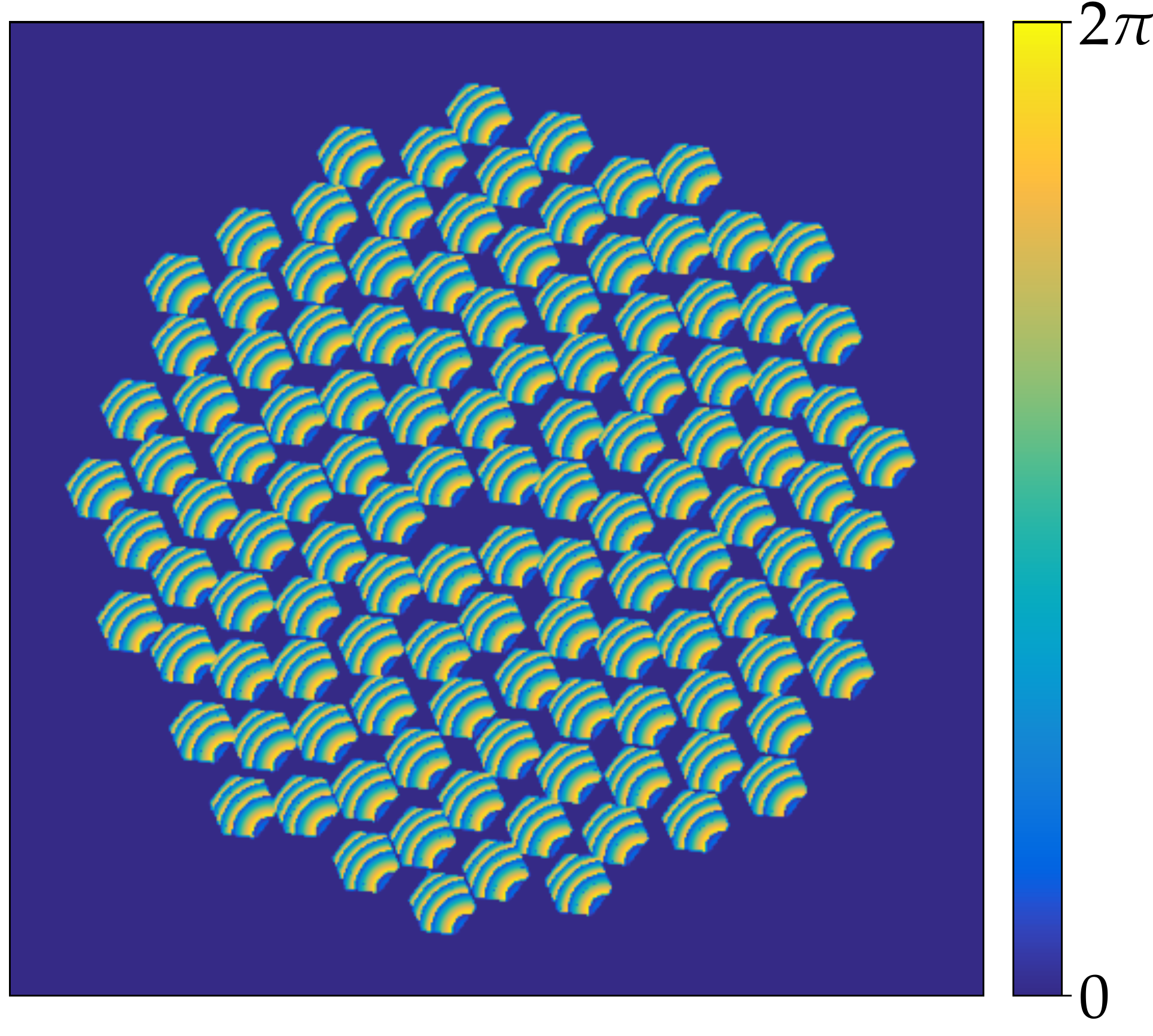}
  \includegraphics[height=.15\textwidth]{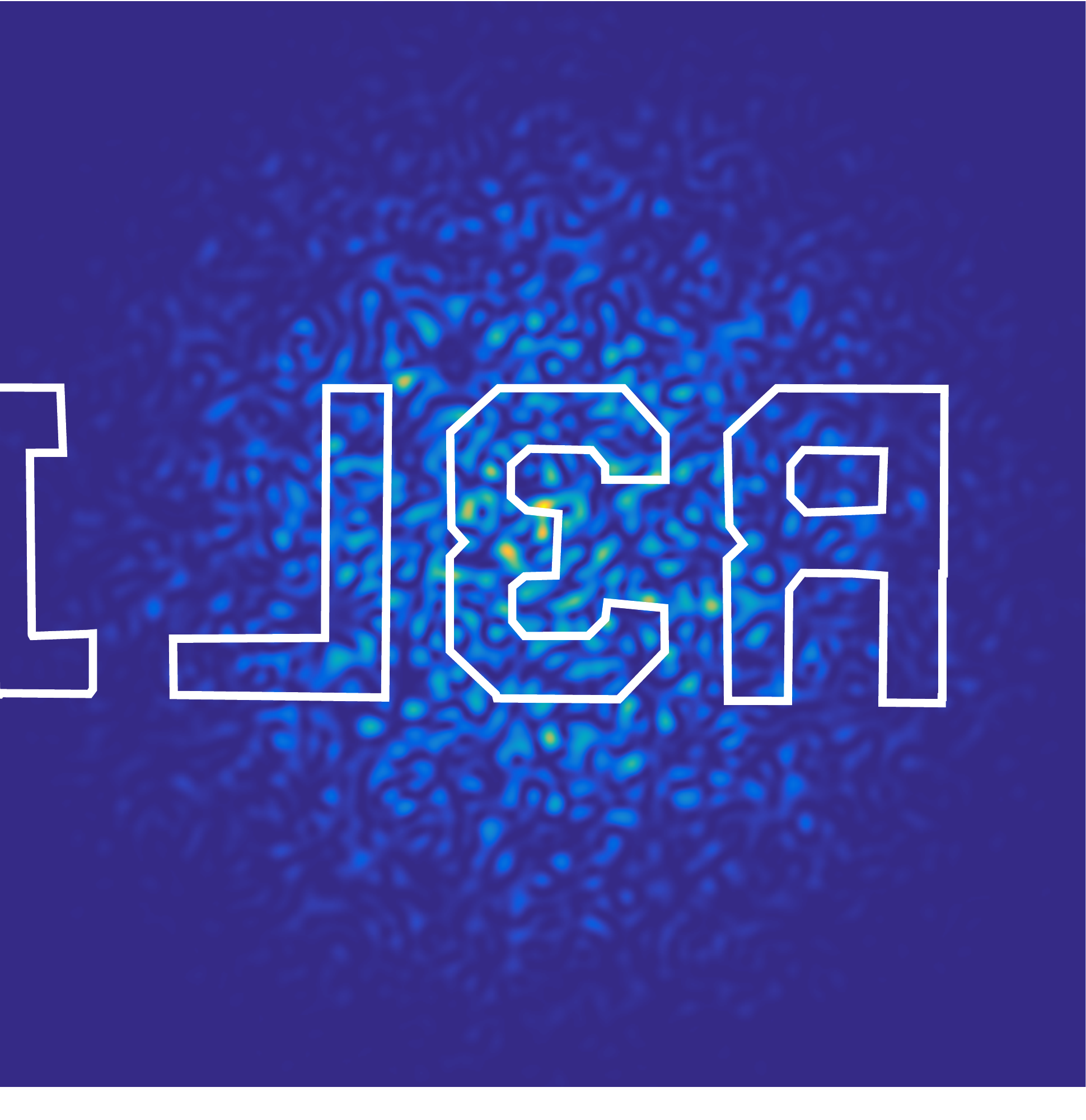}
  \raisebox{.3cm}{\includegraphics[width=.67\textwidth]{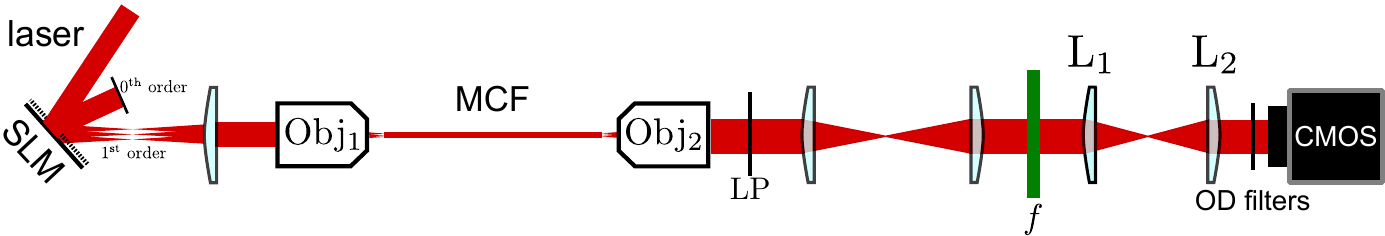}}
  \put(-512,7){\color{white}\small (a)}
  \put(-422,7){\color{white}\small (b)}
  \put(-340,7){(c)}
\caption{(a) SLM configuration ($800\times 600$ pixels) with lenslet hexagonal arrays dedicated to each core. Blaze gratings applied to each microlens deflect the ray beams towards the MCF proximal end while the $0^{\text{th}}$ beam is reflected out of the optical path. (b) Speckle generated from $\bs\alpha=(e^{\im \theta_q})_{q=1}^Q$ with $\theta_q \sim_{\iid} \cl U[0,2\pi]$. The part of the speckle reaching the camera is within the white contour lines representing the studied object $f$. (c) Schematic of the optical setup. Cutoff $\lambda_c=600$nm, SLM=Spatial Light Modulator, MCF=Multi-Core Fiber, 
LP=Linear Polarizer, $f$=object to be imaged, OD=Optical Density (neutral density filters). 
\label{fig:optical_elements}}
\end{figure*}

\section{Experimental MCF-LI} \label{sec: exp}

We have tested our approach on proof-of-concept imaging system set in a transmission mode so as to limit both light power loss and Poisson noise~\cite{Sivankutty2018} on the measurements. We describe below the key aspects of this setup, its specific SLM-to-speckle calibration, before providing examples of reconstructed images and studying the influence of the number of cores and speckle illuminations on the quality of the reconstruction.

\subsection{Setup} \label{sec:setup}

In the setup explained in Fig.~\ref{fig:optical_elements}, a continuous wave laser operating at $\lambda=1053$nm, (YLM-1, IPG Photonics) is expanded and impinges upon a Spatial Light Modulator (SLM X10468-07, Hammamatsu) used to code the incident wavefront to the MCF. The MCF is made of $Q=110$ cores arranged in Fermat’s golden spiral, each exhibiting a single mode at the laser wavelength~\cite{Guerit2021}. The MCF exhibits an inter-core coupling term less than 20 dB~\cite{Sivankutty2018}. Unlike multimode fibers with stronger core coupling, the focused or speckle patterns generated by an MCF are resilient to thermal and mechanical external perturbations.

The SLM consists of a $800\times 600$ grid of liquid-crystal phase modulators that control the phase of reflected light. As shown in Fig.~\ref{fig:optical_elements}(a), by mapping specific pixel groups (segments) on the SLM to individual cores of the fiber, an orthogonal basis of input modes is created to modulate the light entering each core. After calibrating the SLM's phase response, any phase pattern in the range of [$0$-$2\pi$] can be conveniently represented as an 8-bit grayscale image. 
The phase pattern on each segment $q$ comprises three terms \emph{(i)} a blazed grating ensures to shift the modulated light to the first-order of the SLM, preventing unmodulated beam from entering the fibers, \emph{(ii)} a convex lens and a series of telescopes produce a focused spot array aligned with the fiber cores, achieving single-modal behavior with a demagnification factor of 64; and \emph{(iii)} a constant phase-offset for each segment which controls the relative phases between the segments.

The light coming from the SLM is focused into the MCF proximal end by Obj$_1$ ($20\times$/$0.75$NA, Nikon), then re-expanded on the distal end side with Obj$_2$ ($20\times$/$0.45$NA, Olympus). The dependency to fiber bending is avoided by placing the MCF straight. In these conditions, the transmission matrix remains constant and a single calibration phase is sufficient. However, it is possible to greatly attenuate the MCF sensitivity to external perturbation such as mechanical bending by twisting the MCF along the fiber during its manufacture~\cite{Tsvirkun2019}. To ensure the validity of the scalar model described in Sec.~\ref{sec:model-descr}, a linear polarizer is placed after the fiber end to eliminate any polarization effects. 
To satisfy the far-field approximation, the object is positioned at the front-focal plane of a lens while the fiber's distal end is placed at the back-focal plane of the same lens~\cite{Goodman2005}. In our setup, the fiber is positioned at the focal plane of Objective lens (Obj$_2$), and lenses L$_1$ and L$_2$ ($75$ and $100$mm, respectively) are used to re-image the conjugate focal plane to a more accessible location on the optical bench (see Fig.~\ref{fig:optical_elements}). The object can be positioned within $\pm 3.5$mm tolerance, easily achieved with standard positioning equipment.

The conjugate focal plane is re-imaged onto a ${1\,920 \times 1\,200}$ CMOS camera (BFLY-U3-23S6M-C, FLIR) which aids in the calibration and positioning of the system desribed in Sec.~\ref{sec:calib}. The same CMOS camera is also used for emulating single-pixel detection by summing up the pixels of the image. Each measurement has an integration time of $19.2$ms, and Optical Density (OD) filters are applied to match the light intensity to the camera's dynamic range for improved accuracy.
Working in transmission mode, we image a negative 1951 USAF test target mask, contoured in Fig.~\ref{fig:optical_elements}(b). The sample image $f$ is thus binary. 

\subsection{Calibration and generalized sensing model}
\label{sec:calib}

Our MCF-LI setup contains optical system imperfections that are difficult to model. For instance, regarding the hypotheses stated in Sec.~\ref{sec:model-descr}, \emph{(i)} the interferometric matrix should be estimated on a set of continuous, off-grid, visibilities, \emph{(ii)} the imaging depth $z$ is a priori unknown and the fiber core diameters are not constant, and \emph{(iii)} the linear polarizer (see Sec.~\ref{sec:setup}) induces spatially variable, but deterministic, attenuation of the sketching vector component.

Rather than correcting all these deviations one by one, we adopt the generalized MCF-LI sensing introduced in Sec.~\ref{sec:beyond-FF} in~(\ref{eq:single-ROP-LE-general}). This requires us to properly calibrate the system and to determine for each core $q \in [Q]$ the complex wavefields $E_q$ in the object plane $\cl Z$ from intensity-only measurements. We thus follow a standard 8-step phase-shifting interferometry technique~\cite{Cai04}. We first fix a reference core, arbitrarily indexed at $q=0$, and we program the SLM to activate only that core and another core $q$, for $1\leq q \leq Q$. We then record in the CMOS camera the 8 fringe patterns $I_{q0}(\bs x; \phi_k)$ induced by the light interference for 8 different phase steps $\phi_k=\frac{2\pi k}{8}$ ($k \in[8]$) of the reference core, as well as the intensity $I_{00}(\bs x; 0) = r_0^2(\bs x)$ obtained from activating only the reference core. Mathematically, given the polar representation $r_q(\bs x) e^{\im \varphi_q(\bs x)}$ of each complex wavefields,
\begin{align*} 
\begin{split}
    I_{q0}(\bs x; \phi_k) &= \big\lvert r_0(\bs x) e^{\im \varphi_0(\bs x)+\phi_k} + r_q(\bs x) e^{\im \varphi_q(\bs x)} \big\rvert^2 \\ 
    &= I_q^{\rm s}(\bs x) + I_q^{\rm i}(\bs x)
    \cos \big( \varphi_{q0}(\bs x) +\phi_k \big),
\end{split}
\end{align*} 
where $I_q^{\rm s} := r_0^2 + r_q^2$, $I_q^{\rm i} := 2 r_0 r_q$ and $\varphi_{q0}:=\varphi_q-\varphi_0$.
We can then recover $r_q(\bs x)$ and $\varphi_{q0}(\bs x)$ in each $\bs x$ by first applying a $8$-length DFT on $I_{q0}(\bs x; \phi_k)$ along the phase steps, and next dividing the last (7-th) DFT coefficient $4 I_q^{\rm i}(\bs x)e^{\im \varphi_{q0}(\bs x)}$ by $8 r_0(\bs x) = 8 \sqrt{I_{00}(\bs x)}$, which gives 
\begin{equation}
  \label{eq:Eq}
  \tilde E_q(\bs x) 
  =  r_q (\bs x) e^{\im (\varphi_{q}(\bs x) - \varphi_0(\bs x))} = E_q(\bs x) e^{-\im \varphi_0(\bs x)}.
\end{equation}
From fields estimated in~\eqref{eq:Eq} for all $q\in [Q]$, we can reproduce any speckle $S(\bs x;\bs\alpha)$ generated from a sketching vector $\bs \alpha \in \bb C^Q$ using~(\ref{eq:speckle-gen}) since this equation is independent of $e^{-\im \varphi_0(\bs x)}$. 

While the model~(\ref{eq:single-ROP-LE-general}) extends beyond the farfield assumption---it only relies on accurate estimation of the wavefields---the optical constraints followed in Sec.~\ref{sec:setup} to reach the farfield model are necessary. They allow these fields to not strongly deviate from pure complex exponentials, which preserves the validity of the FOV and sampling assumptions~\ref{h:bounded-FOV} and~\ref{h:band-limitedness-fvign} in the sensing model.  

In particular, applying the debiasing procedure explained in Sec.~\ref{sec:debiasing}, we get the debiased observation model 
\begin{equation} \label{eq:exp_model}
    \bs y^{\rm c} = \ILEop(f) + \bs n^{\rm c},
\end{equation}
where $\ILEop(f)$ is now associated with the generalized interferometric matrix $\bs G$ defined in~\eqref{eq:intmat-gen}. In other words, we abuse the notations of~\eqref{eq:debiased-ILEop-sensing-model} and consider a sensing operator $\ILEop: h \mapsto \ILEop(h) := \ropA^{\rm c}(\bs G[h])$ applied to a non-vignetted continuous image~$h$. Regarding the computation of $\ILEop$, we leverage the calibration to compute an estimate $\tilde\ILEop(\bs h):=\ropA^{\rm c}(\tilde{\bs G}[\bs h])$ from a sampling $\bs h \in \bb R^{N^2 \simeq N \times N}$ of $h$, assuming that the proximity to the far-field assumption ensures that $\ILEop(h) \approx \tilde\ILEop(\bs h)$. For each measurement $m \in [M]$, we in fact compute $z_m = \langle \tilde{\bs S}(\cdot;\bs\alpha_m), \bs h \rangle$, with $\tilde{\bs S}(\cdot;\bs\alpha_m)$ computed from the estimated fields in~(\ref{eq:Eq}), before to debiase all measurements from~(\ref{eq:deb-z}), \ie $(\tilde\ILEop(\bs h))_m = z_m^{\rm c}$. Therefore, the matrix $\tilde{\bs G}$ is never explicitly estimated.

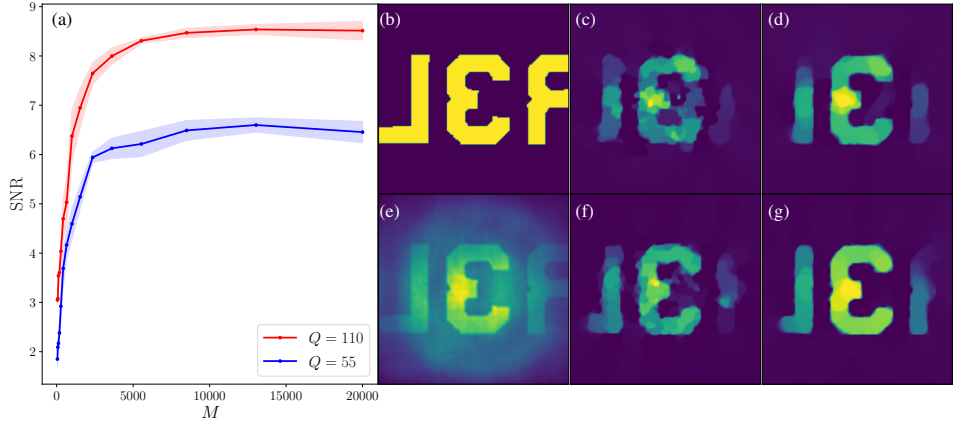
\begin{figure*}[t]
  \centering
  \scalebox{0.7}{
  \includegraphics[width=\textwidth]{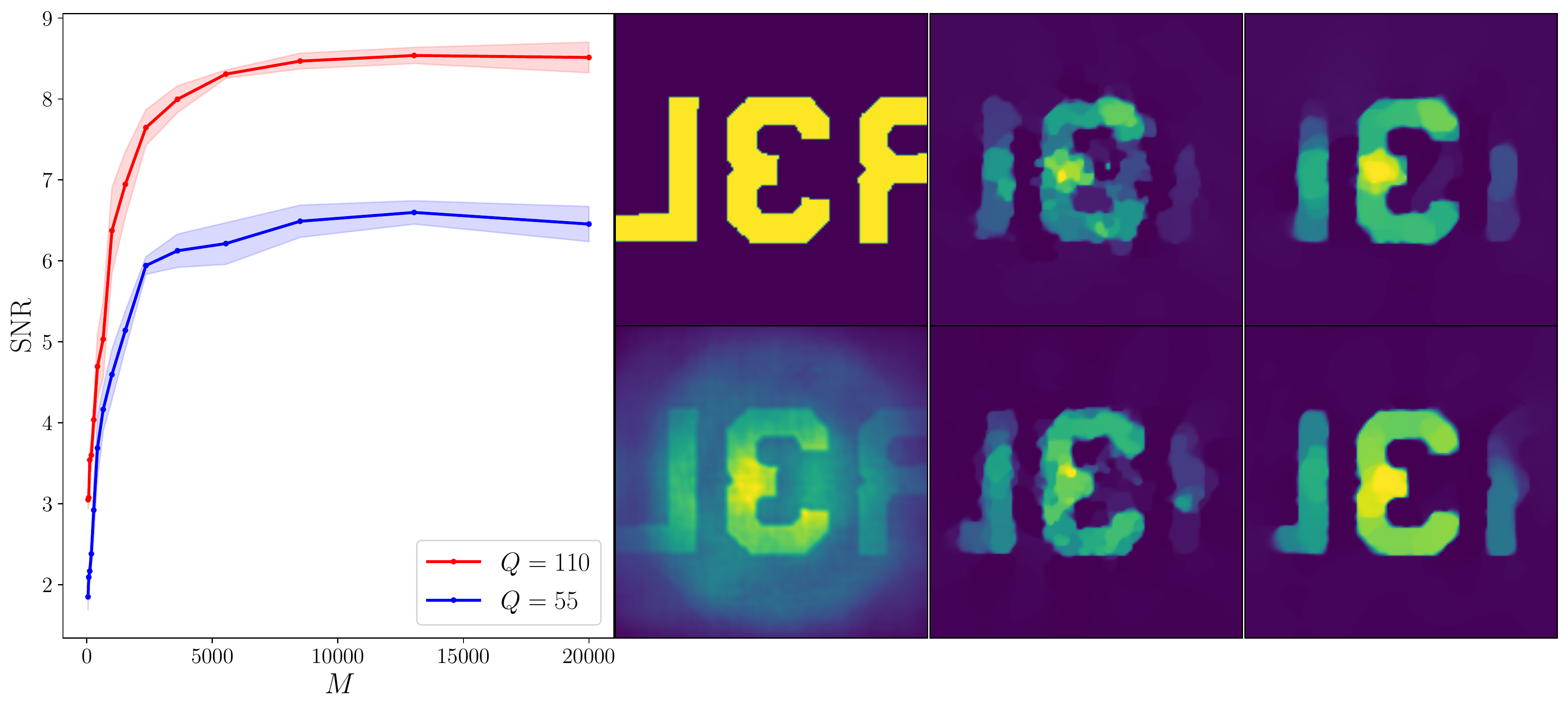}
  \put(-490,218){(a)}
  \put(-313,218){\textcolor{white}{(b)}}
  \put(-207,218){\textcolor{white}{(c)}}
  \put(-103,218){\textcolor{white}{(d)}}
  \put(-313,114){\textcolor{white}{(e)}}
  \put(-207,114){\textcolor{white}{(f)}}
  \put(-103,114){\textcolor{white}{(g)}}}
\caption{Experimental reconstruction on $N=256\times 256$ images. (a) SNR$(\bs w\tilde{\bs f},\bs{wf})$ vs. number of observations $M$ for $Q=55$ (blue) and $Q=110$ (red) cores. Solid lines represent the average, and light areas show $\pm1\sigma$ positions from 5 trials. (b) Ground truth $f$ obtained by illuminating the USAF target with white light passing through the MCF (c-d) Reconstruction using $M=\{49,2\cdot 10^4\}$ with $Q=55$ cores (e) Rec. in RS mode (see Sec.~\ref{sec:prev-mcf-modes}) (f-g) Same as (c-d) with $Q=110$ cores. (b-g) are zoomed-in versions of the camera plane seen in Fig.~\ref{fig:optical_elements}(b). \label{fig:exp_results}}
\end{figure*}

\subsection{Results}

We now present examples of reconstructed sample images obtained with the considered optical setup described in Sec.~\ref{sec:setup}, and the calibration and the sensing model from Sec.~\ref{sec:calib}.

For these experiments, our reconstruction scheme differs from the one followed in Sec.~\ref{sec:image-reconstruction}. First, as explained in~Sec.~\ref{sec:calib}, the sensing model considers a sampling of the un-vignetted sample image $f$, with a sensing operator computed in the pixel domain. Second, instead of the $\ell_1$-prior, we decided to estimate this image by promoting a small total variation (TV) norm, as it is more adapted to the cartoon-shape model of the USAF targets. This deviation from the theoretical setting described in Sec.~\ref{sec:image-reconstruction} aims to provide insights if the proposed SROP model and associated single calibration can also reconstruct images more realistic than sparse images. Third, the non-smooth data fidelity term of~\ref{eq:BPDN} is replaced by a smooth square $\ell_2$-norm to ease the iterative computation of the associated convex optimization. We thus solve the following optimization scheme: 
\begin{equation} \label{eq:GFB}
  \tilde{\bs f} = \argmin_{\bs f}~ \tinv{2M} \norm{\bs y^{\rm c}- \tilde{\ILEop}(\bs f)}{2}^2 + \rho \norm{\bs f}{\text{TV}}~\st~\bs f \geq 0,
\end{equation}
with $\rho$ set empirically to $2\,10^6$. As the vignetting limits the image quality on the frontier of the FOV, we decided to measure the quality of the estimated images with the SNR achieved between the vignetted ground truth $\bs{wf}$ and the vignetted reconstruction $\bs w \odot\tilde{\bs f}$, \ie ${\rm SNR}(\tilde{\bs f}, \bs f) = 20 \log_{10} ({\norm{\bs{w}\odot \bs f}{2}}/{\norm{\bs{w} \odot (\bs f- \tilde{\bs f})}{2}})$ with the estimated vignetting $\bs w:= Q^{-1} \sum_{q=1}^Q |\tilde{\bs E_q}|^2$.

Experimental reconstruction analyses are provided in Fig.~\ref{fig:exp_results} for images of $N=256\times 256$ pixels. In accordance with~\ref{h:sketch-distrib}, the phase of the $Q$ components of the sketch vectors were uniformly drawn \iid in $[0,2\pi]$ with the 8-bit resolution allowed by the SLM. This configuration maximizes the intensity of light injected in the cores. We tested two values for $Q$, $Q=110$, when all the MCF cores are used, and $Q=55$, by downsampling the Fermat's spiral by a factor 2. In Fig.~\ref{fig:exp_results}(a), we tested the quality of the reconstruction for $M \in [49, 20\,000]$. Transitions similar to those in Fig.~\ref{fig:transition}(a) occur for a small number of observations and a plateau is reached around $M=5000$, representing a compression factor of $M/N=7.6\%$. The highest SNR reached with $Q=110$ cores is better than with $Q=55$ cores, as higher image frequencies are captured with the denser core configuration. This effect can also be viewed in Fig.~\ref{fig:exp_results}(c-d,f-g).
Compared to the reconstruction obtained in Fig.~\ref{fig:exp_results}(e) with the RS mode modeled in Sec.~\ref{sec:prev-mcf-modes}, the TV norm penalty reduces the blur of the reconstructed object.
The low SNR values attained in Fig.~\ref{fig:exp_results} are due to the comparison of the reconstructed images with an imperfect ``ground truth'' which is also an estimation of the sample $f$ using white light illumination. 

\section{Conclusion and Perspectives} \label{sec:conclusion}

In this paper, we extended the modeling of MCF-LI with speckle illumination by including the physics of light propagation. This new model highlights that the sensing of a 2-D refractive index map of interest is limited both by the number of applied illuminations and the number (and arrangement) of cores at the distal end of the MCF. We provided recovery guarantees and observed the derived sample complexities in both numerical and experimental conditions.

A future research could open our recovery guarantees in Sec.~\ref{sec:image-reconstruction} to general sparsifying bases $\bs\Psi \neq \Id$. This would require particularizing our proofs to sparse signals with zero mean, \ie belonging to $\Sigma_K^0:=\{ \bs v \in \Sigma_K : \sum_j (\bs \Psi\bs v)_j =0 \}$.

As announced in Sec.~\ref{sec:limitations}, a limitation of our approach lies in the distinct visibility Assumption~\ref{h:distinct-visib}. By construction, the density of the visibilities---as achieved by a difference set---cannot be uniform. As shown in Fig.~\ref{fig:LIMCF}(b), this is also true for the golden Fermat's spiral arrangement. Therefore, when $Q$ grows on a fixed frequency resolution, close visibilities are hardly distinguishable. A more promising sensing model could integrate a variable density sampling (VDS) of the image spectral domain~\cite{Puy11, adcock_17}. In the same time, this could also allow for more general sparsifying basis by accounting for their variable local coherence with the Fourier basis. However, combining this aspect inside the ROP model is an open question.

Future works about MCF-LI include experimental proof of concept in reflective/endoscopic conditions, extension of the model to vector diffraction theory, and imaging of 3-D maps with generalized ROP models.

\section*{Acknowledgment}

The authors thank Y. Wiaux for interesting discussions on the link between MCF-LI and radio-interferometry. Computational resources have been provided by the supercomputing facilities of UCLouvain (CISM) and the Consortium des Equipements de Calcul Intensif en Fédération Wallonie Bruxelles (CECI) funded by FRS-FNRS, Belgium. This project has received funding from FRS-FNRS (Learn2Sense, T.0136.20; QuadSense, T.0160.24) and European Research Council (ERC, SpeckleCARS, 101052911).

\begin{appendices}

\section{Nyquist recovery of the inteferometric matrix} \label{app:nyqu-interfer-reconstr}
The following proposition provides a deterministic scheme to reconstruct, in a noiseless scenario, any interferometric matrix $\intMa \in \cl H^Q$ from $Q(Q-1)+1$ SROPs corresponding to its intrinsic complexity.
\begin{proposition}
  \label{prop:pair2pair_unit}
  There exists a set of $M=Q(Q-1)+1$ sketching vectors 
    $\left\{ \bs \alpha_m \right\}_{m=1}^M \in \Cbb^Q$ such that any Hermitian matrix $\intMa \in \cl H^Q$ with constant diagonal entries can be reconstructed from the $M$ sketches $y_m = \bs \alpha_m^* \intMa 
    \bs \alpha_m$. 
\end{proposition}
\begin{proof} 
Given the $2$-sparse sketching vectors $\bs\alpha_\gamma(q,r) := \bs e_q + \gamma \bs e_r$, with $q,r \in [Q]$, $|\gamma| = 1$ and the $s$-th canonical vector $\bs e_s$,  
we have $h_\gamma[q,r] := \bs\alpha_\gamma^*(q,r) \intMa \bs\alpha_\gamma(q,r) = \intMa[q,q] + \intMa[r,r] + \gamma \intMa[q,r] + \gamma^* \intMa[r,q] = \frac{2}{Q} \tr \intMa + 2 \real{\gamma \intMa[q,r]}$. 
Therefore,
  \begin{equation} \label{eq:offdiag}
      h_1[q,r] + \im h_{-\im}[q,r] = 2 \intMa[q,r] + 
      \tfrac{2}{Q} (1+\im) \tr \intMa. 
    \end{equation}
From the $Q(Q-1)$ sketching vectors $\{\bs\alpha_\gamma(q,r): 1 < q<r\leq Q, \gamma \in \{1,-\im\}\} \subset \bb C^Q$, the value $2 \Re\{H\} = H+H^*$ computed from the sum $H = \sum_{1 < q<r\leq Q} (h_1[q,r] + \im h_{-\im}[q,r])$ respects
$$
\ts \Re\{H\} = \sum_{q \neq r}\intMa[q,r] +  (Q-1) \tr \intMa = \bs 1^\top \intMa \bs 1 + (Q-2) \tr \intMa.
$$
Using the additional unit sketching vector $\bs 1$ thus recovers $\tr \intMa$---and all constant diagonal entries of the Hermitian matrix $\intMa$---from $\Re\{H\}$, and~\eqref{eq:offdiag} provides all its off-diagonal entries. Overall $\intMa$ is thus recovered from $1 + Q(Q-1)$ measurements. 
\end{proof}

\section{Proof of Proposition~\ref{prop:L2L1}.} \label{app:proof_L2L1}

The proof of this proposition is inspired by the one of~\cite[Lemma 2]{chen2015exact}, itself inspired by~\cite{candes2008restricted}. This lemma was developed in the context of sparse matrix recovery from SROP measurements using a variant of BPDN regularized by the trace of the matrix estimate.  While certain elements of our proof are similar to the one of that lemma, its adaptation to the context of sparse signal recovery from~\ref{eq:BPDN} (with a $\ell_1$ fidelity) is not direct, which justifies the following compact derivations.

Let us first write $\tilde{\bs f} = \bs f + \bs h$ with the true image $\bs f$, and some residual $\bs h \in \bb R^N$.
We define the support $T_0 = \supp \bs f_K$ containing the indices of the $K$ strongest entries of $\bs f$. Next, recursively for $1 \leq i \leq \left\lceil (N-K)/K'\right\rceil$, we define the supports $T_i := \supp \big(\bs h_{T_{:i}^c}\big)_{K'}$ of length at most $K'$ containing the indices of the $K'$ strongest entries of $\bs h_{T_{:i}^c}$, with $T_{:i} := \bigcup_{j=0}^{i-1} T_i$, and $T_{:i}^c = [N] \setminus T_{:i}$. 

We first observe that, by construction, for all $j \in T_{i+1}$ with~$i \geq 1$,
$\ts |h_j| \leq \tinv{K'} \sum_{l \in T_i} |h_l| = \tinv{K'} \|\bs h_{T_i}\|_1$ so that 
$\|\bs h_{T_{i+1}}\|^2 \leq \tinv{K'} \|\bs h_{T_i}\|_1^2$. This shows that 
\begin{equation} \label{eq:1}
\ts  \sum_{i \geq 2} \| \bs h_{T_i}\| \leq \tinv{\sqrt {K'}} \sum_{i \geq 1} 
  \| \bs h_{T_i}\|_1 = \tinv{\sqrt{K'}} \| \bs h_{T_0^c}\|_1.
\end{equation}

By optimality of $\tilde{\bs f}$ in~\ref{eq:BPDN} and using twice the triangular inequality, we have
\begin{multline*}
  \|\bs f\|_1 \geq \|\tilde{\bs f}\|_1 = \|\bs f+\bs h\|_1 \geq \|\bs f_{T_0}+ \bs h\|_1 - \|\bs f_{T_0^c}\|_1 \\
  \geq \|\bs f_{T_0}\|_1 + \| \bs h_{T_0^c}\|_1 - \| \bs h_{T_0}\|_1 - \| \bs f_{T_0^c}\|_1,
\end{multline*}
so that 
\begin{equation} \label{eq:2}
\ts \|\bs h_{T_0^c}\|_1 \leq  2 \|\bs f_{T_0^c}\|_1 + \|\bs h_{T_0}\|_1 \leq 2 \|\bs f_{T_0^c}\|_1 + \sqrt{K} \|\bs h_{T_0}\|.
\end{equation}
Therefore, combining~(\ref{eq:2}) and~(\ref{eq:1}) we get
\begin{equation}
  \label{eq:3}
  \ts \sum_{i \geq 2} \|\bs h_{T_i}\|  \le 2 \frac{\|\bs f_{T_0^c}\|_1}{\sqrt{K'}}
  + \frac{\sqrt K}{\sqrt{K'}} \|\bs h_{T_0}\|.
\end{equation}

By linearity of $\ILEop$ and since both $\bs f$ and $\tilde{\bs f}$ are feasible vectors of the~\ref{eq:BPDN} constraint, we note that since $\bs h = \bs f -\tilde{\bs f}$
$$
\|\ILEop(\bs h)\|_1 \leq \|\ILEop(\bs f)-\bs y^{\rm c}\|_1+ \|\ILEop (\tilde{\bs f})-\bs y^{\rm c}\|_1 \leq 2 \epsilon.
$$
Therefore, if $\ILEop$ has the $\text{RIP}_{\ell_2/\ell_1}(\Sigma_k, {\sf m}_{k}, {\sf M}_{k})$ for $k \in \{K', K+K'\}$, we can develop the following inequalities
\begin{align*}
  \ts \frac{2\epsilon}{M} &\geq \tinv M \|\ILEop(\bs h)\|_1 \geq \tinv M 
                 \|\ILEop(\bs h_{T_{:1}})\|_1 - \tinv M \|\ILEop (\bs h_{T_{:1}^c})\|_1 \\
               &\ts \geq {\sf m}_{K+K'} \|\bs h_{T_{:1}}\| - 
                 \tinv M \sum_{i \geq 2} \|\ILEop(\bs h_{T_i})\|_1\\
               &\ts \geq \tinv{\sqrt 2} {\sf m}_{K+K'} (\|\bs h_{T_0}\|+ \|\bs h_{T_1}\|) - 
                 \tinv M \sum_{i \geq 2} \|\ILEop(\bs h_{T_i})\|_1\\
               &\ts \geq \tinv{\sqrt 2} {\sf m}_{K+K'} (\|\bs h_{T_0}\|+ \|\bs h_{T_1}\|) - 
                 {\sf M}_{K'} \sum_{i \geq 2} \|\bs h_{T_i}\|\\
               &\ts \geq \tinv{\sqrt 2} {\sf m}_{K+K'} 
                 (\|\bs h_{T_0}\|+ \|\bs h_{T_1}\|) - {\sf M}_{K'} \frac{2\|\bs f_{T_0^c}\|_1 + \sqrt K \|\bs h_{T_0}\|}{\sqrt{K'}},
\end{align*}
where we used several times the triangular inequality, the fact that $|T_i|=K'$ for $i\geq1$, and~(\ref{eq:3}) in the last inequality. The passage from the second to the third line is due to $\|\bs h_{T_{:1}}\|^2
= \|\bs h_{T_0}\|^2 + \|\bs h_{T_1}\|^2 \geq  (\|\bs h_{T_0}\| + \|\bs h_{T_1}\|)^2/2$.

Therefore, rearranging the terms, and since $\|\bs f_{T_0^c}\|_1 = \|\bs f - \bs f_K\|_1$, we get
\begin{multline} \label{eq:4}
  \ts \frac{2\epsilon}{M} + \ts 2 {\sf M}_{K'} \frac{\|\bs f-\bs f_K\|_1}{\sqrt{K'}} \\*
  \ts \geq ( \tinv{\sqrt 2} {\sf m}_{K+K'} - {\sf M}_{K'} \frac{\sqrt K}{\sqrt{K'}})
  \|\bs h_{T_0}\| + \frac{{\sf m}_{K+K'}}{\sqrt 2}  \|\bs h_{T_1}\| \\
  \ts \geq (\tinv{\sqrt 2} {\sf m}_{K+K'} - {\sf M}_{K'} \frac{\sqrt K}{\sqrt{K'}})(\|\bs h_{T_0}\| + \|\bs h_{T_1}\|).
\end{multline}
Finally, if $\tinv{\sqrt 2} {\sf m}_{K+K'} - {\sf M}_{K'} \frac{\sqrt K}{\sqrt K'} 
    \geq \gamma >0$ and $K' > 2K$, then
\begin{align*}
\|\bs f-\tilde{\bs f}\| &\ts = \|\bs h\| \leq \|\bs h_{T_0}\| + \|\bs h_{T_1}\| + 
        \sum_{i\geq 2}  \|\bs h_{T_i}\| \\
        &\ts \underset{(\ref{eq:3})}{\leq} \|\bs h_{T_0}\| + \|\bs h_{T_1}\| +
        2 \frac{\|\bs f_{T_0^c}\|_1}{\sqrt{K'}} + 
        \frac{\sqrt K}{\sqrt{K'}} \|\bs h_{T_0}\| \\
        &\ts \leq \frac{\sqrt 2+1}{\sqrt 2} (\|\bs h_{T_0}\| + \|\bs h_{T_1}\|) + 
        2 \frac{\|\bs f_{T_0^c}\|_1}{\sqrt{K'}} \\
        &\ts \underset{(\ref{eq:4})}{\leq} \frac{2+\sqrt 2}{\gamma}
        \big(\frac{\epsilon}{M} + {\sf M}_{K'} \frac{\|\bs f-\bs f_K\|_1}{\sqrt{K'}} \big)  
        + 2 \frac{\|\bs f_{T_0^c}\|_1}{\sqrt{K'}}.
\end{align*}
This thus proves the instance optimality~\eqref{eq:bpdn-inst-opt} by taking
$$
\ts C_0 = \frac{2+\sqrt 2}{\gamma} {\sf M}_{K'}+2,\ \text{and}\ D_0 = \frac{2+\sqrt 2}{\gamma}.
$$

\section{Proof of Prop.~\ref{prop:rip-ileop}}
\label{app:proof-prop-rip-ileop}
 
We will need the following lemmata to prove Prop.~\ref{prop:rip-ileop}.
We first need to prove that $\|\ropA^{\rm c}(\intMa)\|_1$, with $\ropA^{\rm c}$ defined in~(\ref{eq:centered-srop-op}) concentrates around its mean. This slightly extends~\cite[Prop.~1]{chen2015exact} where the authors rather proved that the \emph{debiased} operator $\ropA'$---such that, for any matrix $\bs{\cl I}$ and an even number of measurements $M=2M'$, $\ropA'(\intMa)_i := \ropA(\intMa)_{2i+1} - \ropA(\intMa)_{2i}$ for $i \in [M']$---respects the RIP$_{\ell_2/\ell_1}$. This debiasing is introduced to ensure that $\bb E \ropA'(\intMa) = \bs 0$. We show that this is also true for~$\ropA^{\rm c}$.

We first show some useful facts about $\ropA$ and $\ropA^{\rm c}$.
\begin{lemma}[Mean and anisotropy of the SROP operator]
  \label{lem:mean-aniso-srop}
  Given an Hermitian matrix $\intMa \in \cl H^Q$, a zero-mean complex random variable $\alpha$ with $\bb E \alpha^2 = 0$, and bounded second and fourth moments $\bb E |\alpha|^2 = \mu_2$, and $\bb E |\alpha|^4 = \mu_4$, and a set of random vectors $\{\bs \alpha_m\}_{m=1}^M \subset \bb C^M$ with components \iid as $\alpha_{mq} \sim \alpha$ (for $m \in [M]$, $q \in [Q]$), the SROP operator $\ropA$ associated with $\{\bs \alpha_m\}_{m=1}^M$ is such that
  \begin{align} 
    \bb E \,\ropA_m(\intMa)&= \bb E \scp{\bs \alpha_m\bs \alpha_m^*}{\intMa} = \mu_2 \,\tr \intMa,\ \forall m \in [M] \label{eq:SROP-mean}\\
    \tinv{M}\bb E \ropA^* \ropA (\intMa)&= \mu_2^2 \, \intMa + (\mu_4 - 2\mu_2^2)\, \intMa_{\rm d} + \mu_2^2 (\tr \intMa) \Id, \label{eq:SROP-anisotropy}
  \end{align}
  where the operator $\ropA^*$ is the adjoint\footnote{By definition, the adjoint satisfies $\scp{\ropA \bs M}{\bs v}_{\Rbb^N} = \scp{\bs M}{\ropA^* \bs v}_{\Cbb^{Q\times Q}}$.} of $\ropA$ with
  \begin{equation*} \label{eq:adjoint-A}
  \ts  \ropA^*: \bs z \in \bb R^M \mapsto \ropA^*(\bs z) := \sum_{m=1}^M z_m \bs \alpha_m \bs \alpha_m^* \in \cl H^Q,
  \end{equation*}
  and the matrix $\intMa_{\rm d} := \diag(\diag(\intMa))$ zeroes all but the diagonal entries of $\intMa$. Therefore, if $\intMa, \bc J \in \cl H^Q$ with $\intMa$ hollow, then 
$$
\bb E\,\ropA^{\rm c}(\bc J) = 0,\ \bb E \,\ropA(\intMa) = \bs 0,\ \text{and}\ \tinv{M}\bb E \ropA^* \ropA (\intMa) = \mu_2^2 \, \intMa.
$$

\end{lemma}
\begin{proof}
  Eq.~\eqref{eq:SROP-mean} is an immediate consequence of $\bb E \bs \alpha_m \bs \alpha_m^* = \mu_2 \Id$. Regarding~\eqref{eq:SROP-anisotropy}, we first note that $\bb E \ropA^* \ropA \intMa = \bb E \sum_{m=1}^M (\bs \alpha_m^*  \intMa \bs{\alpha}_m) \bs{\alpha}_m \bs{\alpha}_m^* = M \bb E ( \bs{\alpha}^* \intMa \bs{\alpha}) \bs{\alpha} \bs{\alpha}^*$, and for $q,r \in [Q]$, 
  $[\bb E ( \bs{\alpha}^* \intMa \bs{\alpha}) \bs{\alpha} \bs{\alpha}^*]_{qr} = \sum_{j,k=1}^Q \cl I_{j,k} \bb E (\alpha_j^* \alpha_k \alpha_q \alpha_r^*)$.
  
  If $q=r$, then $\bb E (\alpha_j^* \alpha_k \alpha_q \alpha_r^*) = \bb E (\alpha_j^* \alpha_k |\alpha_q|^2)$ is zero if $j \neq k$, $\mu_2^2$ if $j=k\neq q$, and $\mu_4$ if $j=k=q$. Therefore,
  \begin{align*}
    \ts [\bb E ( \bs{\alpha}^* \intMa \bs{\alpha}) \bs{\alpha} \bs{\alpha}^*]_{qq}&\ts = \sum_{j=1}^Q \cl I_{j,j} \bb E (|\alpha_j|^2 |\alpha_q|^2)\\
    &\ts = \mu_2^2 \tr(\intMa) + (\mu_4 - \mu_2^2) \cl I_{qq}.
  \end{align*}
  If $q\neq r$, then $\bb E (\alpha_j^* \alpha_k \alpha_q \alpha_r^*)$ is non-zero only if $j=q$ and $k=r$ (since $\bb E \alpha^2 = 0$ and $\bb E |\alpha|^2 = \mu_2$), in which case it is equal to $\mu_2^2$. Consequently, $[\bb E ( \bs{\alpha}^* \intMa \bs{\alpha}) \bs{\alpha} \bs{\alpha}^*]_{qr} = \mu_2^2 \cl I_{q,r}$. Gathering these identities, we finally find~\eqref{eq:SROP-anisotropy}.
\end{proof}

The next lemma (adapted from~\cite[App. A]{chen2015exact}) relates the expectation of $\|\ropA(\intMa)\|_1$ to the Frobenius norm of hollow matrices $\intMa$; a useful fact for studying below the concentration of $\|\ropA(\intMa)\|_1$.
\begin{lemma}[Controlling the expected SROP $\ell_1$-norm]
  \label{lem:control-srop-l1}
  In the context of Lemma~\ref{lem:mean-aniso-srop}, if the random variable $\alpha$ has unit second moment ($\mu_2 = 1$) and bounded sub-Gaussian norm $\|\alpha\|_{\psi_2} \leq \kappa$ (with $\kappa \geq 1$), then, for any hollow matrix $\intMa \in \cl H^Q$, the random variable $\xi := \bs \alpha^* \intMa \bs \alpha$ is sub-exponential with norm $\|\xi\|_{\psi_1} \leq \kappa^2$, and there exists a value $0<c_\alpha <1$, only depending on the distribution of $\alpha$,  such that
\begin{equation} \label{eq:RIP_in_expectation}
\ts c_\alpha \| \intMa \|_F  \leq \tinv{M} \bb E \|\ropA(\intMa)\|_1 = \bb E |\xi| \leq \| \intMa \|_F.
\end{equation}
\end{lemma}
\begin{proof}
The proof is an easy adaptation of~\cite[App. A]{chen2015exact} to the random variable $\xi = \scp{\bs \alpha\bs \alpha^*}{\intMa}_F = \bs \alpha^* \intMa \bs \alpha$, for $\intMa$ hollow. The constant $c_1$ (Eq.~50) in that work is here set to $1$ since $\ts  (\bb E |\bs \alpha^* \intMa \bs \alpha|)^2 \leq \bb E |\bs \alpha^* \intMa \bs \alpha|^2 = \tinv{M} \bb E \|\ropA(\intMa)\|^2_2 = \|\intMa\|_F^2$. 
\end{proof}

The following lemma leverages the result above to characterize the concentration of $\tinv{M} \|\ropA(\intMa)\|_1$. 
\begin{lemma}[Concentration of SROP in the $\ell_1$-norm] \label{lemma:concentration_ROPs}
In the context of Lemmata~\ref{lem:mean-aniso-srop} and~\ref{lem:control-srop-l1}, given a hollow matrix $\intMa \in \cl H^Q$, there exists a value $0<c_\alpha<1$, only depending on the distribution of $\alpha$,  such that, for $t \geq 0$, with a failure probability smaller than $2 \exp(- cM \min(t^2, t))$, 
  \begin{equation}
  \label{eq:RIP_concentration}
  \ts (c_\alpha - 2t \kappa^2) \| \intMa \|_F  \leq \tinv{M} \|\ropA(\intMa)\|_1 \leq (1+2t \kappa^2) \| \intMa \|_F.
\end{equation}
\end{lemma}
\begin{proof}
We can assume $\|\intMa\|_F = 1$ by homogeneity of~\eqref{eq:RIP_concentration}.  Defining the random variables $\xi_m := \bs \alpha_m^* \intMa \bs \alpha_m$ and $\wt\xi_m := |\xi_m| - \bb E |\xi_m|$ for $m \in [M]$, Lemma~\ref{lem:control-srop-l1} shows that each $\xi_m$ is sub-exponential with $\|\xi_m\|_{\psi_1} \leq \kappa^2$. Moreover,  using the triangular inequality and $\bb E |\xi_m| \leq \|\xi_m\|_{\psi_1}$ (from~\cite[Def. 5.13]{Vershynin2010}), we get $\|\wt\xi_m\|_{\psi_1} \leq \|\xi_m\|_{\psi_1} + \bb E |\xi_m| \leq 2 \kappa^2$, showing the sub-exponentiality of each $\wt\xi_m$ for $m \in [M]$. 

Therefore, given $t\geq 0$, using~\cite[Cor. 5.17]{Vershynin2010}, we get, with a failure probability lower than $2\exp(-cM \min(\frac{t^2}{4 \kappa^4}, \frac{t}{2\kappa^2}))$,
$$
\ts -t\ \leq \tinv{M} \sum_{m=1}^M \wt\xi_m = \tinv{M} \|\ropA(\intMa)\|_1 - \tinv{M} \bb E \|\ropA(\intMa)\|_1\ \leq t
$$
for some $c > 0$. The result follows by applying~\eqref{eq:RIP_in_expectation} to lower and upper bound $\tinv{M} \bb E \|\ropA(\intMa)\|_1$, followed by a rescaling in $t$.
\end{proof}

Despite the non-independence of the centered matrices $\bs A^{\rm c}_m$ defining the components of $\ropA^{\rm c}$, we can show the concentration of $\ropA^{\rm c}(\bc J)$ in the $\ell_1$-norm by noting that, if~\ref{h:sketch-distrib} holds, $\ropA_m^{\rm c}(\bc J)=\ropA_m^{\rm c}(\bc J_{\rm h})=\ropA_m(\bc J_{\rm h}) - \scp{\bs A^{\rm a}}{\bc J_{\rm h}}$, applying Lemma~\ref{lemma:concentration_ROPs} on the $\ell_1$-norm of the first term, and noting the second concentrates around~0.

\begin{lemma}[Concentration of centered SROP in the $\ell_1$-norm] \label{lemma:concentration_ROPs_for_Ac}
In the context of Lemmata~\ref{lem:mean-aniso-srop} and~\ref{lem:control-srop-l1} and supposing~\ref{h:sketch-distrib} holds, given a matrix $\bc J \in \cl H^Q$ and $\bc J_{\rm h} = \bc J - \bc J_{\rm d}$, there exists a value $0<c_\alpha <1$, only depending on the distribution of $\alpha$,  such that, for $t \geq 0$, with a failure probability smaller than $2 \exp(- cM \min(t^2, t))$, 
  \begin{equation}
  \label{eq:RIP_concentration_Ac}
  \ts (c_\alpha - 3t \kappa^2) \| \bc J_{\rm h} \|_F  \leq \tinv{M} \|\ropA^{\rm c}(\bc J)\|_1 \leq (1+3t \kappa^2) \| \bc J_{\rm h} \|_F.
\end{equation}
\end{lemma}
\begin{proof}
Given $\bc J \in \cl H^Q$ and its hollow part $\bc J_{\rm h} = \bc J - \bc J_{\rm d}$, the operator $\ropA^{\rm c}$ is defined componentwise by $\ropA^{\rm c}_m(\bc J) = \ropA_m(\bc J) - \scp{\bs A^{\rm a}}{\bc J} = \scp{\bs \alpha_m \bs \alpha_m^* - \bs A^{\rm a}}{\bc J} $, with $\bs A^{\rm a}=\frac{1}{M} \sum_{j=1}^M \bs \alpha_j \bs \alpha_j^*$. Moreover, from~\ref{h:sketch-distrib}, $\ropA^{\rm c}_m(\bc J) = \ropA^{\rm c}_m(\bc J_{\rm h})$ since both matrices $\bs \alpha_m \bs \alpha_m^*$ and $\bs A^{\rm a}$ have unit diagonal entries. Therefore, by triangular inequality 
\begin{equation} \label{eq:srop-centered-bounds}
\ts  \big|\tinv{M} \|\ropA^{\rm c}(\bc J)\|_1 - \tinv{M} \|\ropA(\bc J_{\rm h})\|_1\big| \leq |\scp{\bs A^{\rm a}}{\bc J_{\rm h}}|.
\end{equation}
Given the \iid random variables $\xi_j = \bs \alpha_j^* \bc J_{\rm h}\bs \alpha_j$, we get $\scp{\bs A^{\rm a}}{\bc J_{\rm h}} = \frac{1}{M} \sum_{j=1}^M \xi_j$, with $\bb E \xi_j = 0$ from the hollowness of $\bc J_{\rm h}$. According to Lemma~\ref{lem:control-srop-l1}, each $\xi_j$ is sub-exponential with $\|\xi_j\|_{\psi_1} \leq \kappa^2$. Therefore, using again~\cite[Cor. 5.17]{Vershynin2010}, we have, with a failure probability lower than $2\exp(-cM \min(\frac{t^2}{\kappa^4}, \frac{t}{\kappa^2}))$,
$$
\ts -t\ \leq \scp{\frac{1}{M} \sum_{j=1}^M \bs \alpha_j \bs \alpha_j^*}{\bc J_{\rm h}} \leq t,
$$
for some $c > 0$. The result follows from a union bound on the failure of this event and the event~\eqref{eq:RIP_concentration} in Lemma~\ref{lemma:concentration_ROPs}, both inequalities and~\eqref{eq:srop-centered-bounds} justifying~\eqref{eq:RIP_concentration_Ac}.
\end{proof}
As a simple corollary of the previous lemma, we can now establish the concentration of 
$\ILEop(\bs f) := \varpi \ropA^{\rm c}\big(\cl T (\bs F \bs f)\big) \in\bb R^M_+$
in the $\ell_1$-norm for an arbitrary $K$-sparse vector $\bs f \in \Sigma_K$.
\begin{corollary}[Concentration of $\ILEop$ in the $\ell_1$-norm] \label{cor:concentration_B}
In the context of Lemma~\ref{lemma:concentration_ROPs_for_Ac}, suppose that~\ref{h:bounded-FOV}-~\ref{h:sketch-distrib} are respected, with~\ref{h:rip-visibility} set with sparsity level $\splev >0$ and distortion $\delta = 1/2$.  Given $\bs f \in \Sigma_\splev$, and the operator $\ILEop$ defined in~(\ref{eq:ileop-def}) from the $M$ SROP measurements and the $|\cl V_0|=Q(Q-1)$ non-zero visibilities with
$$
Q(Q-1) \geq 4 \splev\, {\rm plog}(N,\splev, \delta),
$$
we have, with a failure probability smaller than $2 \exp(- c' M)$ (for some $c'>0$ depending only on the distribution of $\alpha$),  
\begin{equation*}
\label{eq:RIP_concentration-for-ilerop}
\ts \frac{\varpi c_\alpha}{2\sqrt 2}\frac{\sqrt{|\cl V_0}|}{\sqrt{N}}\, \| \bs f \|  \leq \tinv{M} \|\ILEop(\bs f)\|_1 \leq 2 \varpi \frac{\sqrt{|\cl V_0}|}{\sqrt{N}}\, \| \bs f\|.
\end{equation*} 
\end{corollary}
\begin{proof} 
  Given $\bs f \in \Sigma_\splev$ and $\bc J = \cl T (\bs F \bs f) \in \cl H^Q$, let us assume that~\eqref{eq:RIP_concentration_Ac} holds on this matrix with $t=c_\alpha/(6\kappa^2) < 1/6$, an event with probability of failure smaller than $2 \exp(- c' M)$ with $c'>0$ depending only on $c_\alpha$ and $\kappa$, \ie on the distribution of~$\alpha$. We first note that $\|\bc J_{\rm h}\|_F = \|\bs R_{\overline{\cl V}_0} \bs F \bs f\|$ from~\eqref{eq:equiv-frob-l2}. Second,
\begin{equation}
  \ts \tinv{2} \|\bs f\|^2 \leq \tfrac{N}{|\cl V_0|}\|\bs R_{\overline{\cl V}_0} \bs F \bs f\|^2 \leq \frac{3}{2} \|\bs f\|^2.\label{eq:tmp-proof-ilerop-concent}
\end{equation}
since from~\ref{h:rip-visibility} the matrix $\bs \Phi := \sqrt{N} \bs R_{\overline{\cl V}_0} \bs F$ respects the RIP$_{\ell_2/\ell_2}(\Sigma_\splev,\delta=1/2)$ as soon as $|\cl V_0| = Q(Q-1) \geq 4 \splev\, {\rm plog}(N,\splev, \delta)$.
  Therefore, since $\ILEop(\bs f) = \varpi\ropA^{\rm c}(\bc J) = \varpi\ropA^{\rm c}(\bc J_{\rm h})$, combining~(\ref{eq:RIP_concentration_Ac}) and~\eqref{eq:tmp-proof-ilerop-concent} gives
  \begin{multline*}
    \ts \tinv{M} \|\ILEop(\bs f)\|_1 \geq (c_\alpha -3t\kappa^2) \varpi \|\bc J_{\rm h}\|_F\\
\ts = \tinv{2} c_\alpha \varpi \|\bs R_{\overline{\cl V}_0} \bs F\bs f\| \geq \frac{\varpi c_\alpha}{2\sqrt 2}\,\frac{\sqrt{|\cl V_0|}}{\sqrt{N}}\,\|\bs f\|.
  \end{multline*}
  Similarly, using $\sqrt{\frac{3}{2}}(1+3t\kappa^2) < (\frac{3}{2})^{3/2}<2$, we get
  $$
  \ts \tinv{M} \|\ILEop(\bs f)\|_1 \leq \sqrt{\frac{3}{2}}(1+3t\kappa^2)\varpi\,\frac{\sqrt{\cl V_0}}{\sqrt{N}}\,\|\bs f\| < 2\varpi\frac{\sqrt{|\cl V_0}|}{\sqrt{N}}\,\|\bs f\|.
  $$
\end{proof}

We are now ready to prove Prop.~\ref{prop:rip-ileop}. We will follow the standard proof strategy developed in~\cite{Baraniuk2008}. By homogeneity of the RIP$_{\ell_2/\ell_1}$ in~\eqref {eq:RIP-L2L1-def}, we restrict the proof to unit vectors $\bs f$ of $\Sigma_\splev$, \ie $\bs f \in \Sigma_\splev^* := \Sigma_\splev \cap \bb S^{N-1}_2$. 

Given a radius $0<\lambda < 1$, let $\cl G_\lambda \subset \Sigma^*_\splev$ be a $\lambda$ covering of $\Sigma^*_\splev$, \ie for all $\bs f \in \Sigma^*_\splev$, there exists a $\bs f' \in \cl G_\lambda$, with $\supp \bs f' = \supp \bs f$, such that $\|\bs f - \bs f'\| \leq \lambda$. Such a covering exists and its cardinality is smaller than ${N \choose \splev}(1 + \frac{2}{\lambda})^{\splev} \leq (\frac{3eN}{\splev\lambda})^\splev$~\cite{Baraniuk2008}.

Invoking Cor.~\ref{cor:concentration_B}, we can apply the union bound to all points of the covering so that
\begin{equation}
\label{eq:UB-concent-ileop}
  \ts \ \forall \bs f' \in \cl G_\lambda,\   \ts \frac{\varpi c_\alpha}{2\sqrt 2}\frac{\sqrt{|\cl V_0|}}{\sqrt N}  \leq \tinv{M} \|\ILEop(\bs f')\|_1 \leq 2 \varpi \frac{\sqrt{|\cl V_0|}}{\sqrt N},
\end{equation}
holds with failure probability smaller than
$$
\ts 2 (\frac{3eN}{\splev\lambda})^\splev \exp(- c' M ) \leq 2 \exp(\splev \ln(\frac{3eN}{\splev\lambda}) - c' M). 
$$
Therefore, there exists a constant $C > 0$ such that, if $M \geq C \splev \ln(\frac{3eN}{\splev\lambda})$, then~\eqref{eq:UB-concent-ileop} holds with probability exceeding $1 - 2 \exp(- c'' M)$, for some $c'' > 0$.

Let us assume that this event holds. Then, for any $\bs f \in \Sigma_\splev$, 
\begin{align*}
  \ts \tinv{M} \|\ILEop(\bs f)\|_1&\ts \leq \tinv{M} \|\ILEop(\bs f')\|_1 + \tinv{M} \|\ILEop(\bs f - \bs f')\|_1\\
                                  &\ts \leq 2 \varpi \frac{\sqrt{|\cl V_0|}}{\sqrt N} + \tinv{M} \|\ILEop(\frac{\bs f - \bs f'}{\|\bs f - \bs f'\|})\|_1 \|\bs f - \bs f'\|\\
                                  &\ts \leq 2 \varpi \frac{\sqrt{|\cl V_0|}}{\sqrt N} + \tinv{M} \|\ILEop(\bs r)\|_1 \lambda,
\end{align*}
with the unit vector $\bs r := \frac{\bs f - \bs f'}{\|\bs f - \bs f'\|}$.  However, this vector $\bs r$ is itself $\splev$-sparse since $\bs f$ and $\bs f'$ share the same support. Therefore, applying recursively the same argument on the last term above, and using the fact that $\|\ILEop(\bs w)\|_1$ is bounded for any unit vector $\bs w$, we get
$\tinv{M} \|\ILEop(\bs r)\|_1\lambda \leq 2 \varpi \frac{\sqrt{|\cl V_0|}}{\sqrt N} \sum_{j\geq 1} \lambda^j = 2\frac{\lambda}{1-\lambda} \varpi \frac{\sqrt{|\cl V_0|}}{\sqrt N}$.

Consequently, since we also have 
\begin{align*}
  \ts \tinv{M} \|\ILEop(\bs f)\|_1&\ts \geq \tinv{M} \|\ILEop(\bs f')\|_1 - \tinv{M} \|\ILEop(\bs f - \bs f')\|_1\\
                                  &\ts \geq \frac{\varpi c_\alpha}{2\sqrt 2} \frac{\sqrt{|\cl V_0|}}{\sqrt N} - \tinv{M} \|\ILEop(\bs r)\|_1 \lambda,
\end{align*}
we conclude that 
$$
\ts \frac{\varpi c_\alpha}{2\sqrt 2} (\frac{1-2\lambda}{1-\lambda})\frac{\sqrt{|\cl V_0|}}{\sqrt N} \leq \tinv{M} \|\ILEop(\bs f)\|_1 \leq 2 \varpi \frac{1}{1-\lambda}\frac{\sqrt{|\cl V_0|}}{\sqrt N},
$$
Picking $\lambda = 1/4$ finally shows that, under the conditions described above, $\ILEop$ respects the RIP$_{\ell_2/\ell_1}(\Sigma_\splev,m_\splev,M_\splev)$ with $m_\splev > \frac{\varpi c_\alpha}{3\sqrt 2}\frac{\sqrt{|\cl V_0|}}{\sqrt N}$, and  $M_\splev < \frac{8\varpi}{3}\frac{\sqrt{|\cl V_0|}}{\sqrt N}$.
\end{appendices}

\begin{IEEEbiography}[{\includegraphics[width=1in,height=1.25in,clip,keepaspectratio]{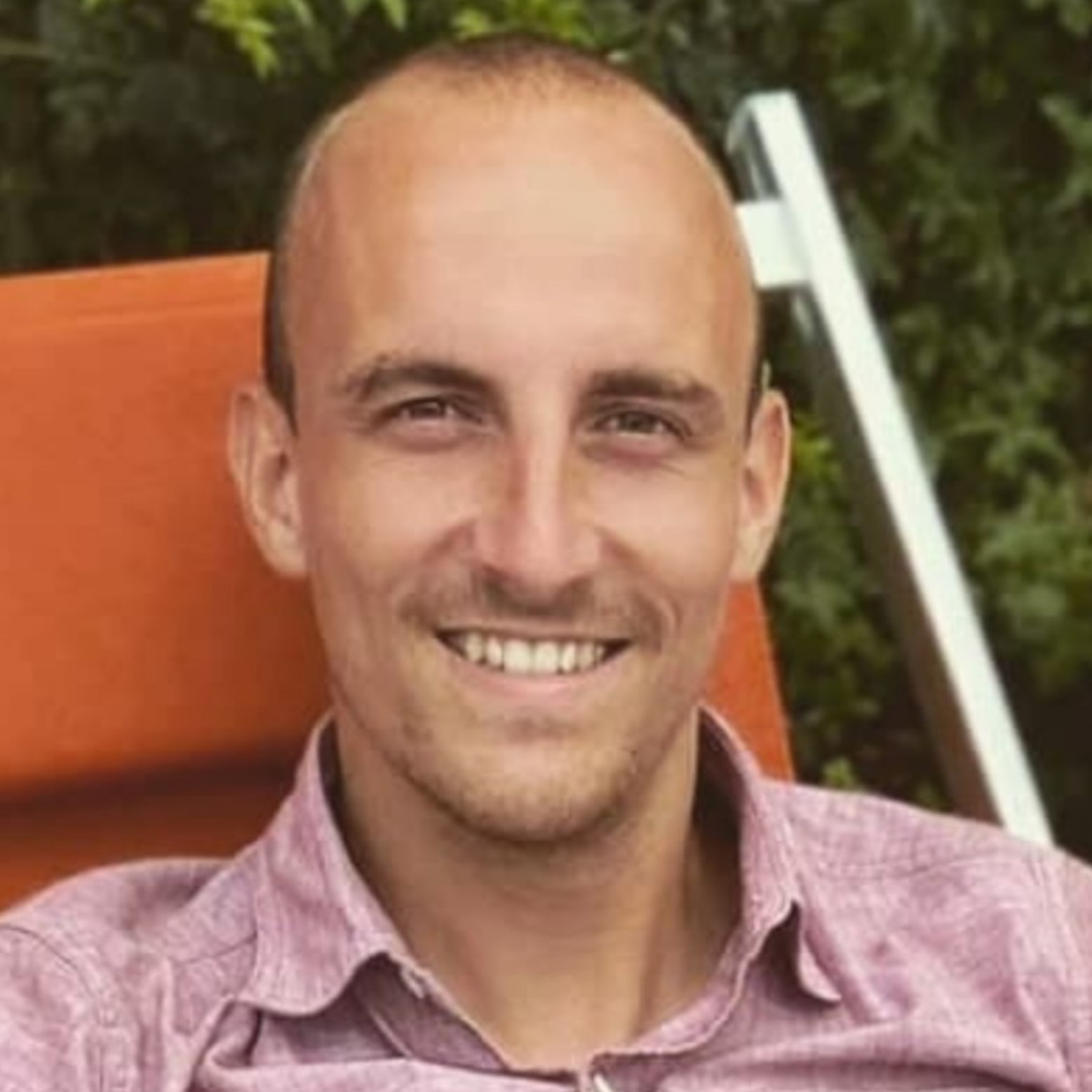}}]{Olivier Leblanc} received the B.Sc. and M.Sc. degrees in electrical engineering from the Mathematical Engineering (INMA) department, ICTEAM/UCLouvain, Louvain-la-Neuve, Belgium, in 2018 and 2020, respectively. His research interests include interferometric imaging and compressive sensing. He is a research fellow funded by the ``Fonds de la Recherche Scientifique'', under the supervision of Laurent Jacques.
\end{IEEEbiography}

\begin{IEEEbiography}[{\includegraphics[width=1in,height=1.25in,clip,keepaspectratio]{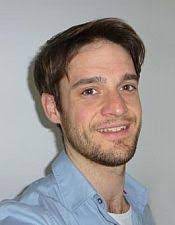}}]{Matthias Hofer} received his M.Sc. degree in electrical engineering and information technology from the Karlsruhe Institute of Technology (KIT), Germany, in 2016, and his PhD degree in physics from Aix-Marseille Université, France, in 2019. He has authored and co-authored peer-reviewed publications in the domains of nonlinear microscopy, complex media and endoscopy.
\end{IEEEbiography}

\begin{IEEEbiography}[{\includegraphics[width=1in,height=1.25in,clip,keepaspectratio]{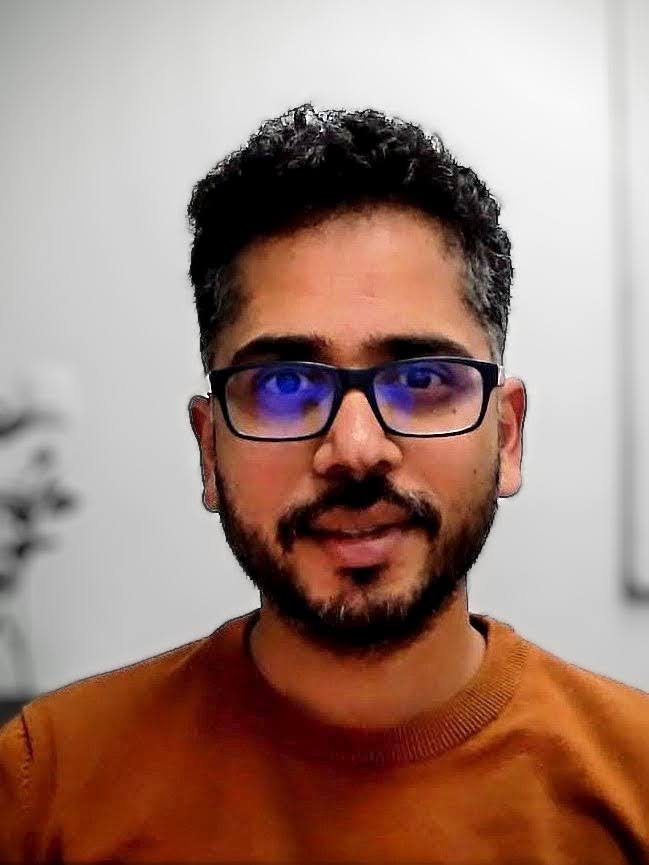}}]{Siddharth Sivankutty} received his dual MSc degree in photonics and applied physics from Friedrich-Schiller Universität, Germany, and the Institut d'Optique and Ecole Polytechnique, France, in 2010, and his PhD in physics from the Université Paris Sud, France, in 2014. Since 2022, he has been a CNRS researcher at the PhLAM lab, Université de Lille, developing nonlinear optical systems for information processing. He is a co-inventor on two patents and has authored more than 10 peer-reviewed publications on lensless endoscopes.
\end{IEEEbiography}

\begin{IEEEbiography}[{\includegraphics[width=1in,height=1.25in,clip,keepaspectratio]{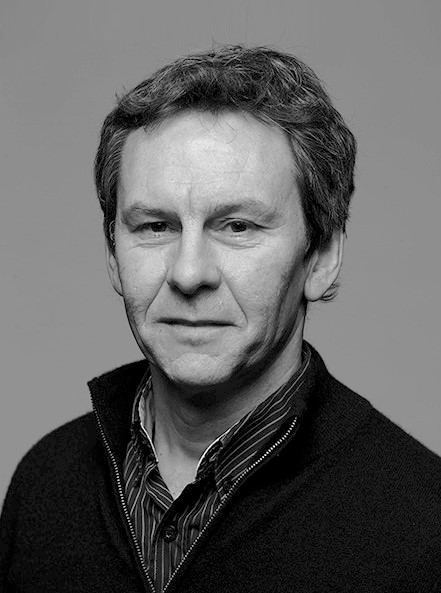}}]{Hervé Rigneault} received his engineering degree from the Ecole Nationale Supérieure de Physique de Marseille (1991) and his MSc and PhD degrees in optical engineering from Aix-Marseille Université (1994). He joined the Centre National de la Recherche Scientifique (CNRS) in 1994 at the Institut Fresnel, where he is currently conducting his research. He is currently a CNRS research director at the Institut Fresnel that develops innovative optical tools for life sciences. He has authored more than 230 peer-reviewed journal papers and is a coinventor of 17 patents in the field of optical engineering and molecular spectroscopy. He has been developing the lensless endoscope at the Institut Fresnel since early 2011.
\end{IEEEbiography}

\begin{IEEEbiography}[{\includegraphics[width=1in,height=1.25in,clip,keepaspectratio]{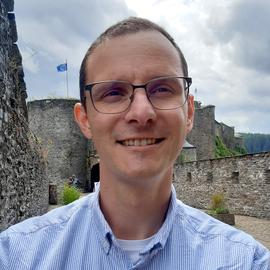}}]{Laurent Jacques} received the B.Sc., M.Sc., and Ph.D. degrees in mathematical physics in 1996, 1998, and 2004, respectively. He has been a FNRS Research Associate from 2012 till 2022 with Image and Signal Processing Group, ICTEAM, UCLouvain, Louvain-La-Neuve, Belgium, and is now professor at the same place. His research interests include sparse signal representations, quantized compressive sensing theory, and computational imaging (\url{https://laurentjacques.gitlab.io} for more information).
\end{IEEEbiography}

\end{document}